\documentclass{sig-alternate}

\pdfoutput=1

\usepackage{alltt}
\usepackage{cite}
\usepackage{color}
\usepackage{comment}
\usepackage{enumerate}
\usepackage{euscript}
\usepackage{graphicx}
\usepackage{ifthen}
\usepackage{microtype}
\usepackage{enumitem}
\usepackage{multirow}
\usepackage{url}
\usepackage{xspace}

\newtheorem{theorem}{Theorem}
\newtheorem{lemma}{Lemma}
\newtheorem{corollary}{Corollary}

\newtheorem{observation}{Observation}

\def\eps{\epsilon}
\def\fr{\frac}
\def\-{\mbox{-}}

\def\real{\mathbb{R}}

\def\nn{\nonumber}

\def\*{\star}

\DeclareMathOperator*{\intr}{\cap}

\def\figcapup{\vspace{-1mm}}
\def\figcapdown{\vspace{-2mm}}
\def\extraspacing{\vspace{2mm} \noindent}

\def\high{\mathit{high}}
\def\highend{\mathit{highend}}
\def\MAX{\mathit{MAX}}
\def\lmax{\mathit{LMAX}}
\def\eps{\epsilon}
\def\fr{\frac}
\def\leftdom{\mathit{leftdom}}
\def\rmax{\mathit{RMAX}}
\def\intr{\cap}
\def\T{\mathcal{T}}

\newboolean{arxivversion}
\setboolean{arxivversion}{true}
\newcommand{\arxivexcl}[2]{\ifthenelse{\boolean{arxivversion}}{#1}{#2}}

\newboolean{confversion}
\setboolean{confversion}{false}
\ifthenelse{\boolean{confversion}}{\includecomment{confenv}}{\excludecomment{confenv}}
\newcommand{\confcmt}[1]{\ifthenelse{\boolean{confversion}}{#1}{}}

\newboolean{fullversion}
\setboolean{fullversion}{true}
\ifthenelse{\boolean{fullversion}}{\includecomment{fullenv}}{\excludecomment{fullenv}}
\newcommand{\fullcmt}[1]{\ifthenelse{\boolean{fullversion}}{#1}{}}

\newcommand{\front}{\text{front}}
\newcommand{\first}{\text{first}}
\newcommand{\records}{\text{records}}
\newcommand{\last}{\text{last}}
\newcommand{\rest}{\text{rest}}
\newcommand{\bigO}{\mathcal{O}}
\newcommand{\smallO}{o}
\newcommand{\nil}{\textit{nil}}
\newcommand{\iref}[1]{\ref{#1}} 
\newcommand{\attr}[1]{\widetilde{#1}}
\renewcommand{\(}{\left(}
\renewcommand{\)}{\right)}

\urldef{\arxivurl}\url{http://arxiv.org/abs/1207.2341}

\usepackage{xcolor}
\usepackage{transparent}
\usepackage{graphicx}
\usepackage{import}

\newcommand{\executeiffilenewer}[3]{
 \ifnum\pdfstrcmp{\pdffilemoddate{#1}}
 {\pdffilemoddate{#2}}>0
 {\immediate\write18{#3}}\fi
}

\newcommand{\yufeigraphics}[2]{
  \arxivexcl
   {\includegraphics[#1]{./#2}}
   {\includegraphics[#1]{./figure/#2}}
}

\usepackage{pifont} 

\begin{document}
\conferenceinfo{PODS'13,} {June 22--27, 2013, New York, New York, USA.}
\CopyrightYear{2013}
\crdata{978-1-4503-2066-5/13/06}
\clubpenalty=10000
\widowpenalty = 10000

\fullcmt{\title{I/O-Efficient Planar Range Skyline \\and Attrition Priority
Queues\thanks{\scriptsize This is the full version of our PODS 2013 paper
with the same title.}}}

\confcmt{\title{I/O-Efficient Planar Range Skyline \\and Attrition Priority
Queues\thanks{\scriptsize The full version is found on http://arxiv.org
under the same title.}}}

\numberofauthors{1}
\author{
	\begin{tabular}{c}
      Casper Kejlberg-Rasmussen$^1$ \hspace{5mm} Yufei Tao$^{2,3}$ \hspace{5mm}
      Konstantinos Tsakalidis$^4$ \\[2mm] Kostas Tsichlas$^5$ \hspace{5mm}
      Jeonghun Yoon$^3$ \\[3mm]
    \affaddr{$^1$MADALGO\thanks{\scriptsize MADALGO (Center for Massive Data
    Algorithmics -- a Center of the Danish National Research Foundation)},
  Aarhus University} \\
		\affaddr{$^2$Chinese University of Hong Kong} \\
		\affaddr{$^3$Korea Advanced Institute of Science and Technology} \\
		\affaddr{$^4$Hong Kong University of Science and Technology} \\
		\affaddr{$^5$Aristotle University of Thessaloniki}
	\end{tabular}
}

\date{}
\maketitle

\begin{confenv}
\begin{abstract}
We study the static and dynamic \emph{planar range skyline reporting problem} in
the external memory model with block size $B$, under a linear space budget.  The
problem asks for an $O(n/B)$ space data structure that stores $n$ points in the
plane, and supports reporting the $k$ maximal input points (a.k.a.\
\emph{skyline}) among the points that lie within a given query rectangle $Q =
[\alpha_1, \alpha_2] \times [\beta_1, \beta_2]$. When $Q$ is \emph{3-sided},
i.e. one of its edges is grounded, two variants arise: \emph{top-open} for
$\beta_2 = \infty$ and \emph{left-open} for $\alpha_1 = -\infty$ (symmetrically
\emph{bottom-open} and \emph{right-open}) queries.

We present optimal static data structures for \emph{top-open} queries, for the
cases where the universe is $\mathbb{R}^2$, a $U \times U$ grid, and rank space
$[\bigO(n)]^2$. We also show that \emph{left-open} queries are harder, as they
require $\Omega((n/B)^\epsilon + k/B)$ I/Os for $\epsilon > 0$, when only linear
space is allowed. We show that the lower bound is tight, by a structure that
supports 4-sided queries in matching complexities.  Interestingly, these lower
and upper bounds coincide with those of the {\em planar orthogonal range
reporting problem}, i.e., the skyline requirement does not alter the problem
difficulty at all!

Finally, we present the first dynamic linear space data structure that supports
top-open queries in $O(\log_{2B^\epsilon} n + k/B^{1-\epsilon})$ and updates in
$O(\log_{2B^\epsilon} n )$ worst case I/Os, for $\epsilon \in [0,1]$.  This also
yields a linear space data structure for 4-sided queries with optimal query I/Os
and $\bigO(\log (n/B))$ amortized update I/Os.  We consider of independent
interest the main component of our dynamic structures, a new real-time
I/O-efficient and catenable variant of the fundamental structure {\em priority
queue with attrition} by Sundar.
\end{abstract}
\end{confenv}

\begin{fullenv}
\begin{abstract}
In the \emph{planar range skyline reporting problem}, the goal is to store a set
$P$ of $n$ 2D points in a structure such that, given a query rectangle $Q =
[\alpha_1, \alpha_2] \times [\beta_1, \beta_2]$, the maxima (a.k.a.\
\emph{skyline}) of $P$ $\cap$ $Q$ can be reported efficiently. The query is
\emph{3-sided} if an edge of $Q$ is grounded, giving rise to two variants:
\emph{top-open} ($\beta_2 = \infty$) and \emph{left-open} ($\alpha_1 = -\infty$)
(symmetrically \emph{bottom-open} and \emph{right-open}) queries.

This paper presents comprehensive results in external memory under the
$\bigO(n/B)$ space budget ($B$ is the block size), covering both the static and
dynamic settings:
\begin{itemize}
  \item For static $P$, we give structures that answer {\em top-open} queries in
    $\bigO(\log_B n + k/B)$, $\bigO(\log\log_B U + k/B)$, and $\bigO(1 + k/B)$
    I/Os when the universe is $\mathbb{R}^2$, a $U \times U$ grid, and a rank
    space grid $[\bigO(n)]^2$, respectively (where $k$ is the number of reported
    points). The query complexity is optimal in all cases.

  \item We show that the {\em left-open} case is harder, such that any
    linear-size structure must incur $\Omega((n/B)^\epsilon + k/B)$ I/Os to
    answer a query. In fact, this case turns out to be just as difficult as the
    general 4-sided queries, for which we provide a static structure with the
    optimal query cost $\bigO((n/B)^\epsilon + k/B)$.

  \item We present a dynamic structure that supports top-open queries in
    $\bigO(\log_{2B^\epsilon} (n/B) + k/B^{1-\epsilon})$ I/Os, and updates in
    $\bigO(\log_{2B^\epsilon}(n/B))$ I/Os, for any $\epsilon$ satisfying $0 \le
    \epsilon \le 1$. This result also leads to a dynamic structure for 4-sided
    queries with optimal query cost $\bigO((n/B)^\epsilon + k/B)$, and amortized
    update cost $\bigO(\log (n/B))$.
\end{itemize}

As a contribution of independent interest, we propose an I/O-efficient version
of the fundamental structure {\em priority queue with attrition} (PQA). Our PQA
supports \textsc{FindMin}, \textsc{DeleteMin}, and \textsc{InsertAndAttrite} all
in $\bigO(1)$ worst case I/Os, and $\bigO(1/B)$ amortized I/Os per operation.
Furthermore, it allows the additional \textsc{CatenateAndAttrite} operation that
merges two PQAs in $\bigO(1)$ worst case and $\bigO(1/B)$ amortized I/Os. The
last operation is a non-trivial extension to the classic PQA of Sundar, even in
internal memory.
\end{abstract}
\end{fullenv}




\section{Introduction} \label{sec:intro}

Given two different points $p=(x_p, y_p)$ and $q=(x_q, y_q)$ in $\real^2$, where
$\real$ denotes the real domain, we say that $p$ \emph{dominates}~$q$ if $x_p
\geq x_q$ and $y_p \geq y_q$. Let $P$ be a set of $n$ points in $\real^2$. A
point $p \in P$ is \emph{maximal} if it is not dominated by any other point in
$P$. The {\em skyline} of $P$ consists of all maximal points of $P$. Notice that
the skyline naturally forms an orthogonal staircase where increasing
$x$-coordinates imply decreasing $y$-coordinates. Figure~\ref{fig:intro-sky}a
shows an example where the maximal points are in black.

\begin{figure}[b]
	\centering
	\begin{tabular}{cc}
    \yufeigraphics{height=25mm}{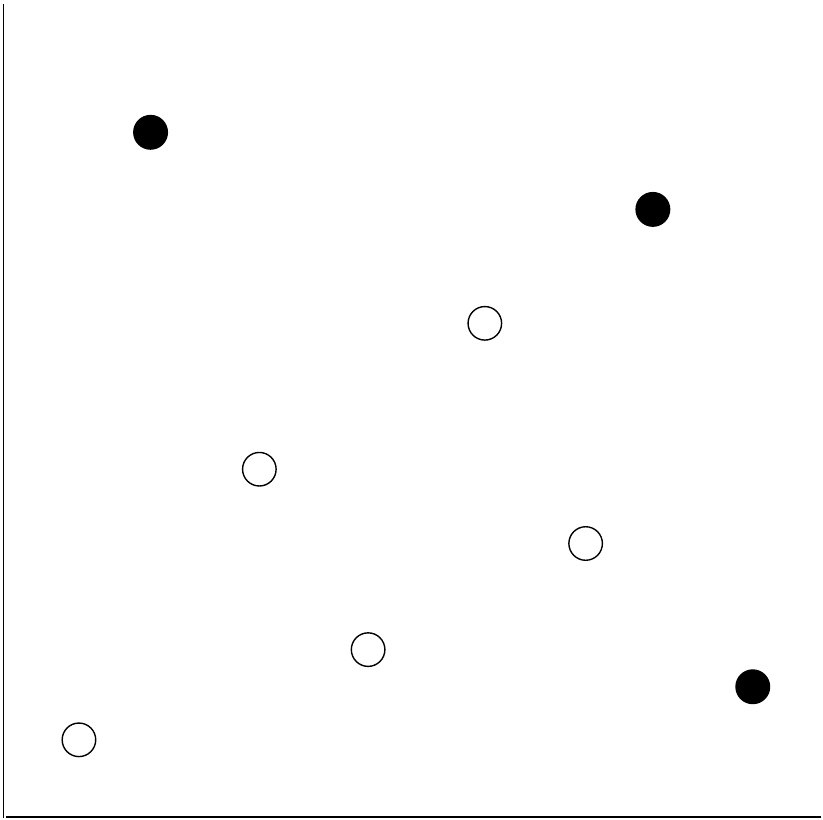}&
    \hspace{3mm}\yufeigraphics{height=25mm}{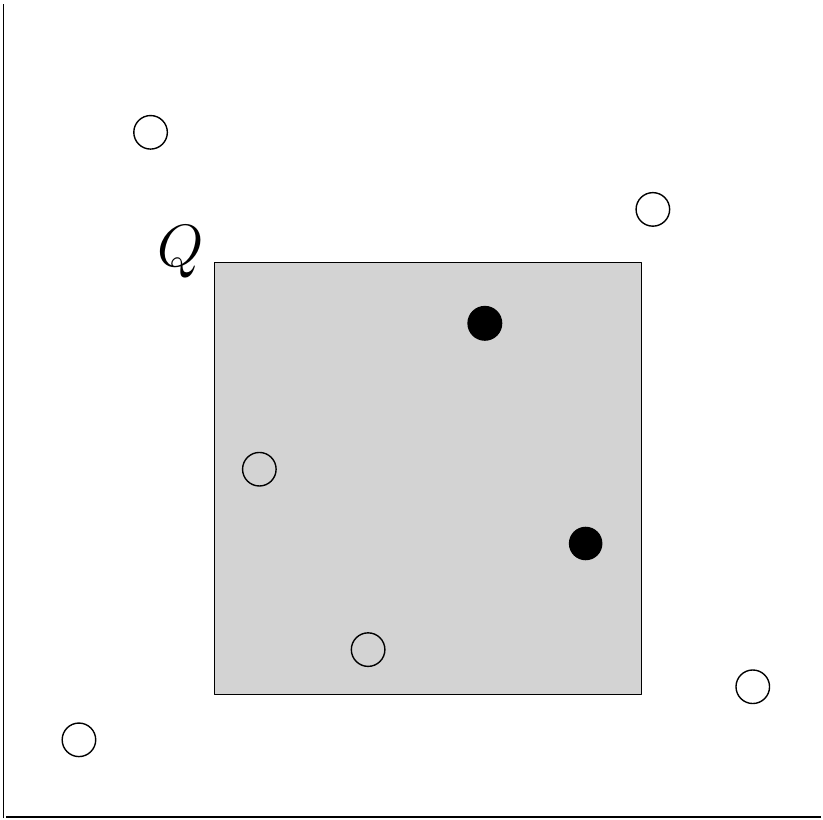} \\
		(a) Skyline &
		\hspace{3mm} (b) Range skyline
	\end{tabular}
	\figcapup
	\caption{Range skyline queries.} \label{fig:intro-sky}
	\figcapdown
\end{figure}

\begin{figure*}
	\centering
	\begin{tabular}{ccccccc}
    \yufeigraphics{height=15mm}{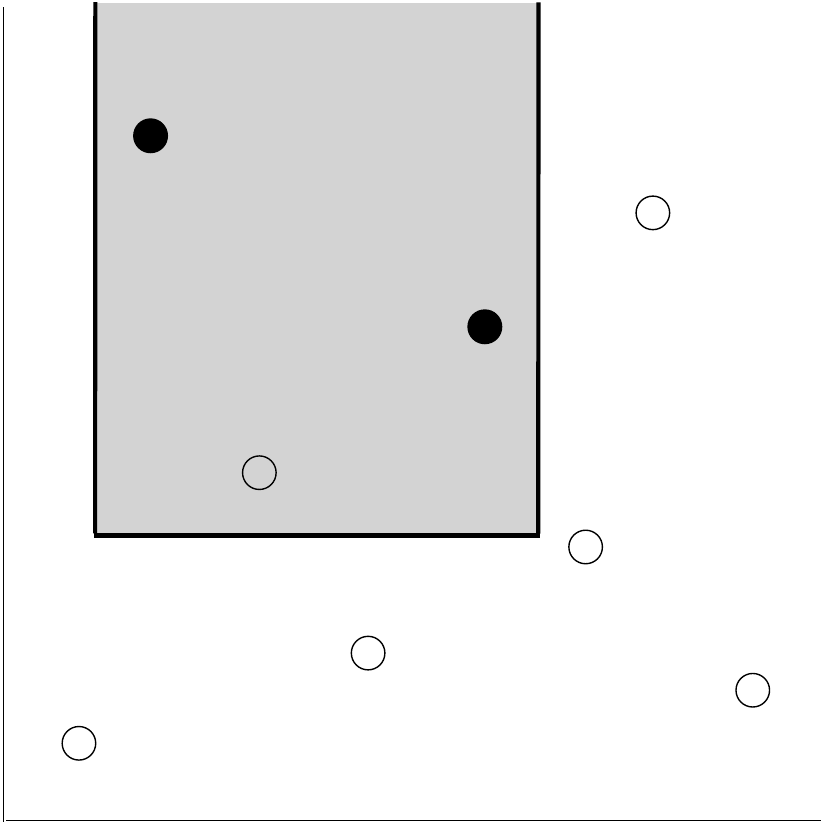}&
    \hspace{-2mm} \yufeigraphics{height=15mm}{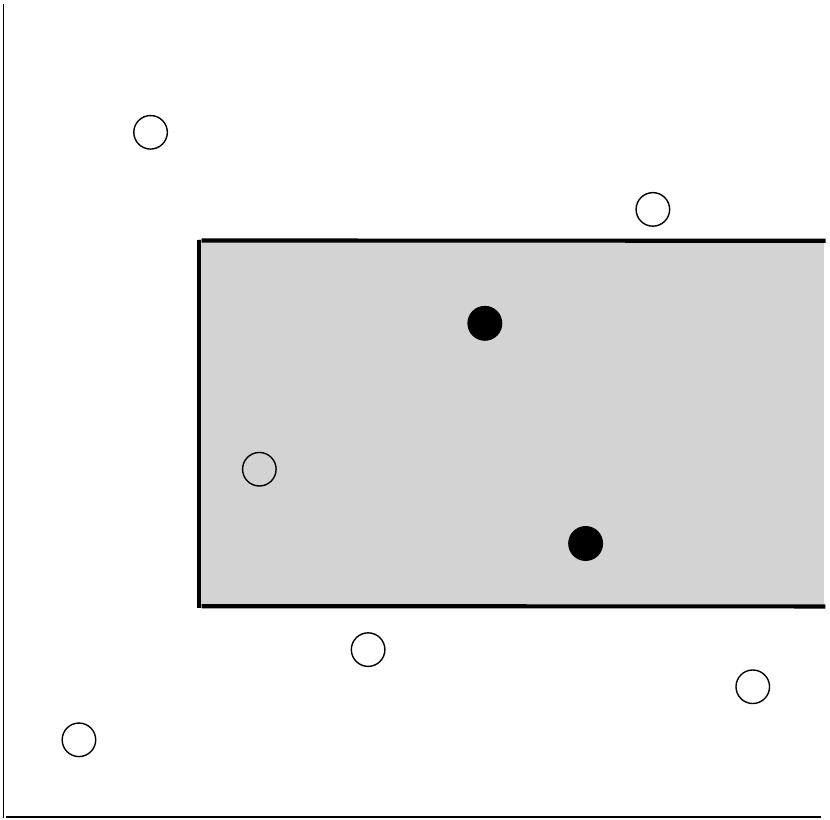}&
    \hspace{-2mm} \yufeigraphics{height=15mm}{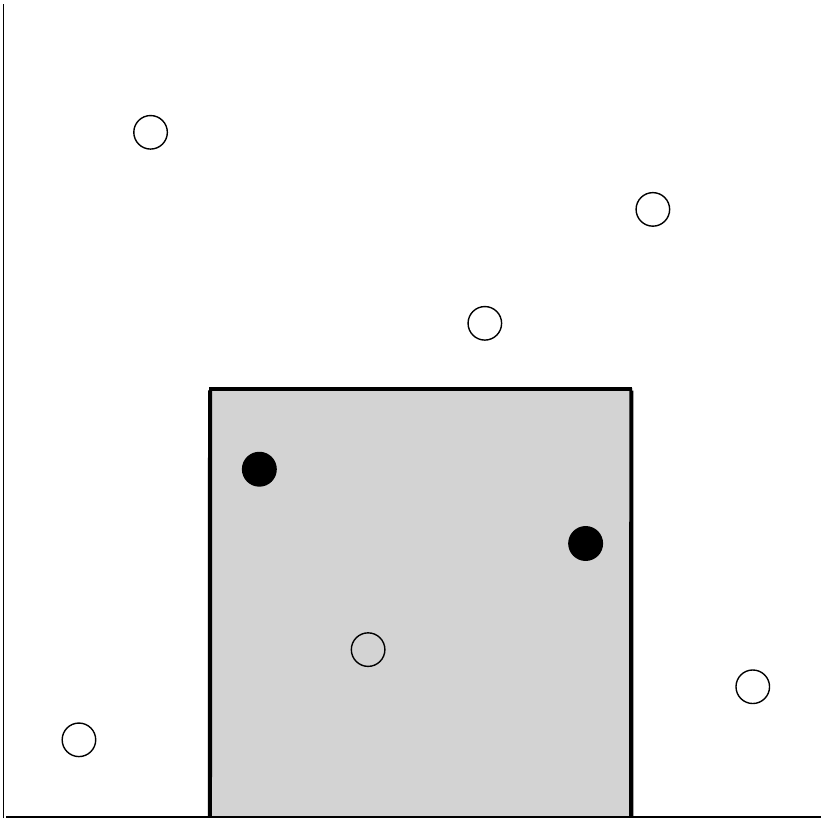}&
    \hspace{-2mm} \yufeigraphics{height=15mm}{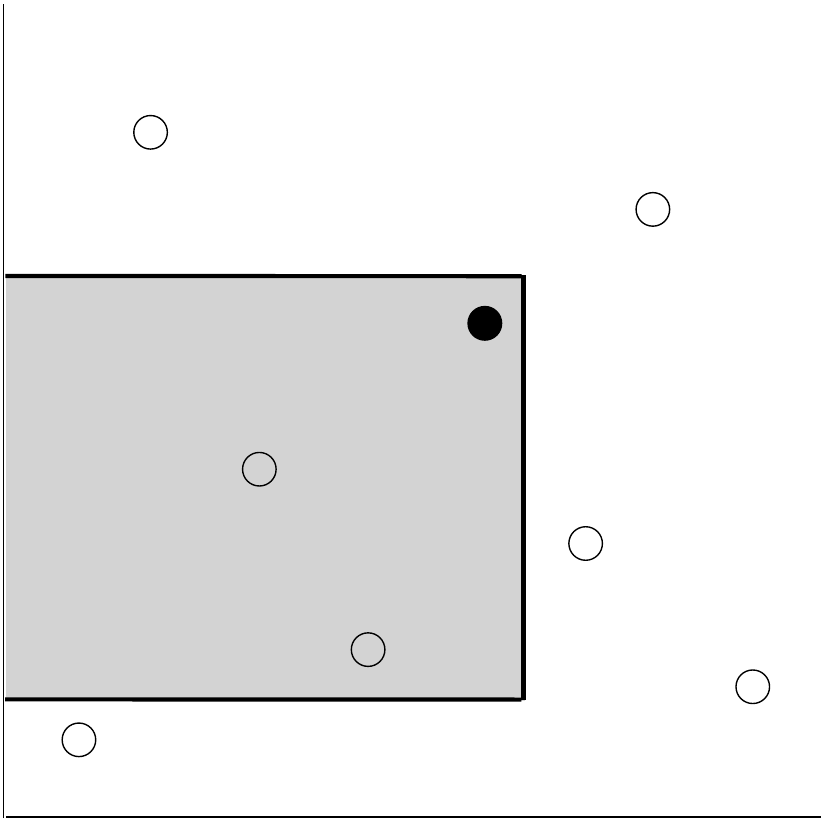} &
    \hspace{-2mm} \yufeigraphics{height=15mm}{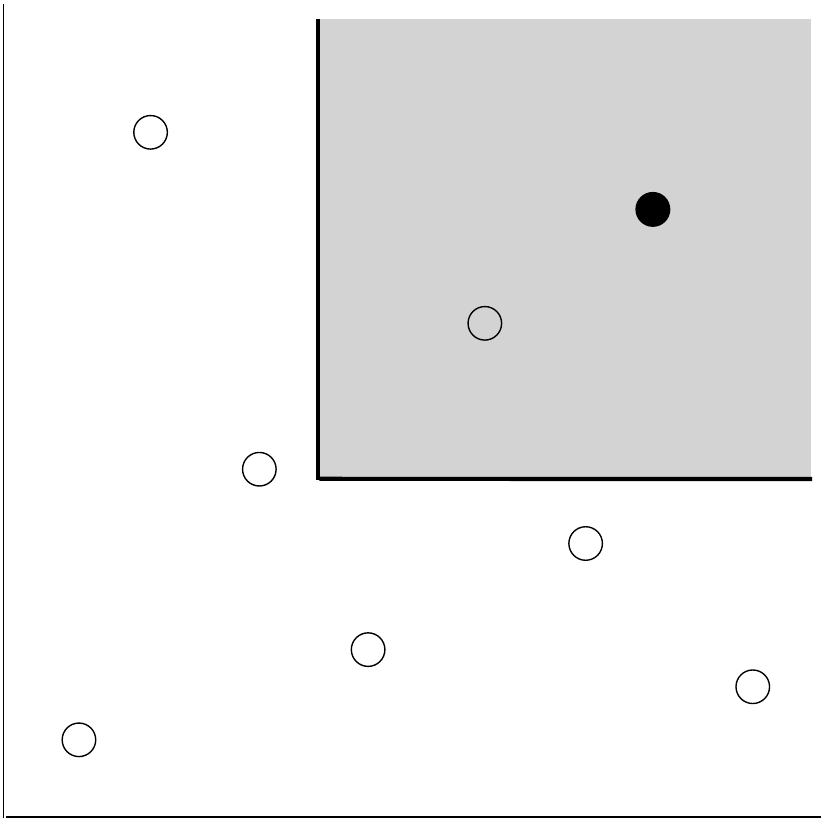} &
    \hspace{-2mm} \yufeigraphics{height=15mm}{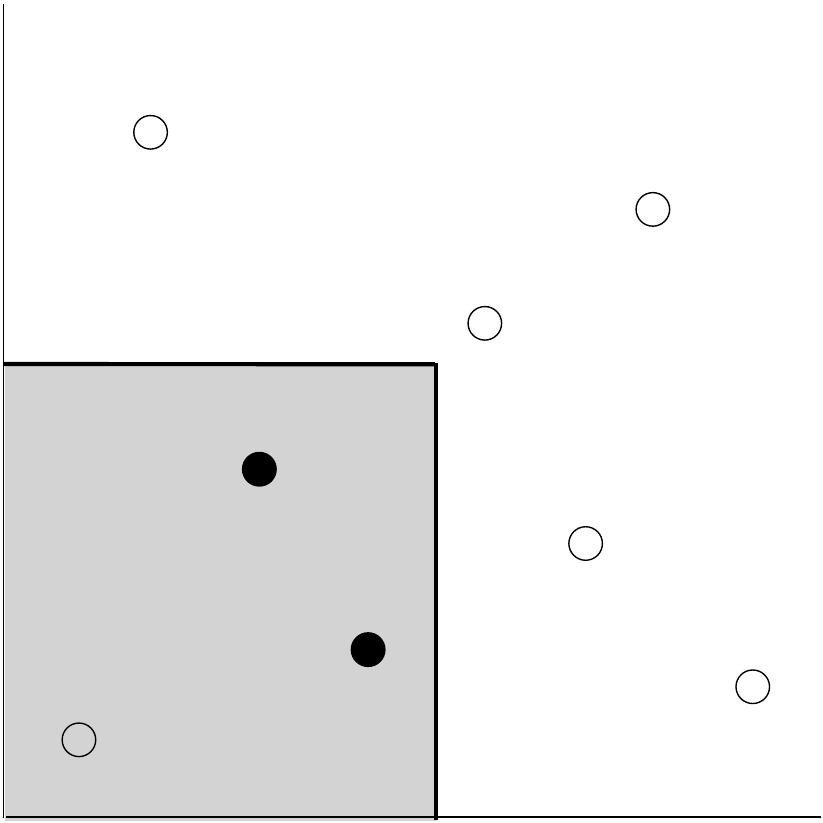}&
    \hspace{-2mm} \yufeigraphics{height=15mm}{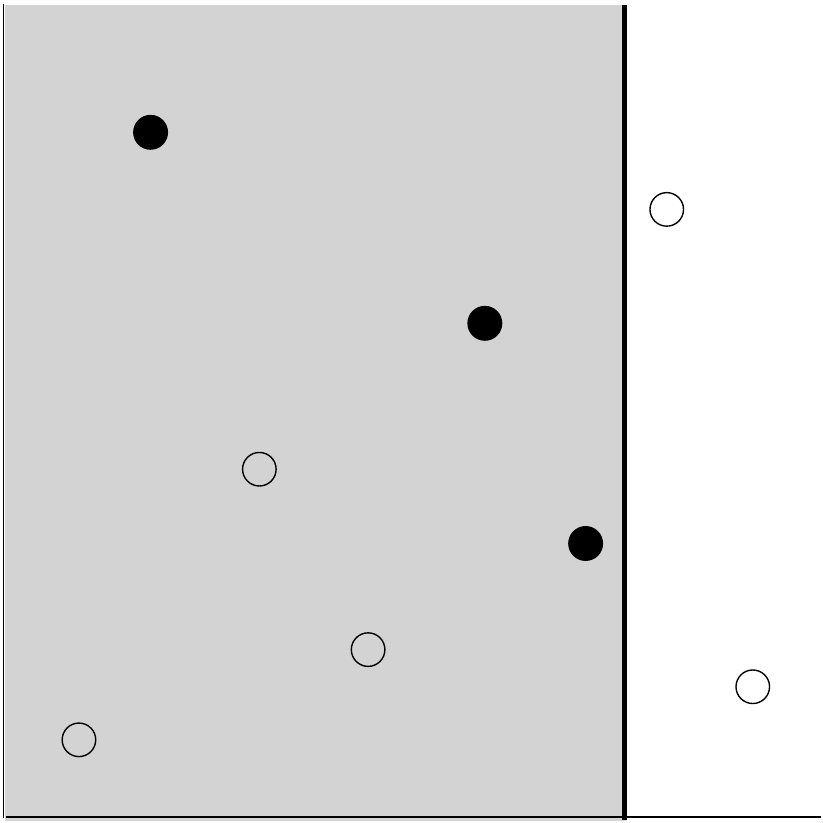} \\
		(a) Top-open &
		\hspace{-2mm} (b) Right-open &
		\hspace{-2mm} (c) Bottom-open &
		\hspace{-2mm} (d) Left-open &
		\hspace{-2mm} (e) Dominance &
		\hspace{-2mm} (f) Anti-dominance &
		\hspace{-2mm} (g) Contour
	\end{tabular}
	\figcapup
	\caption{Variations of range skyline queries (black points represent the query results).} \label{fig:intro-var}
	\figcapdown
\end{figure*}

Given an axis-parallel rectangle~$Q$, a \emph{range skyline query} (also known as a \emph{range maxima query}) reports the skyline of $P \cap Q$.  In Figure~\ref{fig:intro-sky}b, for instance,~$Q$ is the shaded rectangle, and the two black points constitute the query result. When~$Q$ is a 3-sided rectangle, a range skyline query becomes a {\em top-open}, {\em right-open}, {\em bottom-open} or {\em left-open} query, as shown in
Figures~\ref{fig:intro-var}a-\ref{fig:intro-var}d respectively. A {\em dominance} (resp.\ {\em anti-dominance}) query $Q$ is a 2-sided rectangle with both the top and right (resp.\ the bottom and left) edges grounded, as shown in Figure~\ref{fig:intro-var}e (resp.\ \ref{fig:intro-var}f). Another well-studied variation is the {\em contour} query, where $Q$ is a 1-sided rectangle that is the half-plane to the left of a vertical line (Figure~\ref{fig:intro-var}g).

This paper studies linear-size data structures that can answer range skyline
queries efficiently, in both the static and dynamic settings. Our analysis
focuses on the {\em external memory} (EM) model \cite{AV88}, which has become
the dominant computation model for studying I/O-efficient algorithms. In this
model, a machine has $M$ words of memory, and a disk of an unbounded size. The
disk is divided into disjoint {\em blocks}, each of which is formed by $B$
consecutive words. An {\em I/O} loads a block of data from the disk to memory,
or conversely, writes $B$ words from memory to a disk block. The space of a
structure equals the number of blocks it occupies, while the cost of an
algorithm equals the number of I/Os it performs. CPU time is for free.

By default, the data universe is $\real^2$. Given an integer $U > 0$, $[U]$
represents the set $\{0, 1, \ldots, U-1\}$. All the above queries remain well
defined in the universe $[U]^2$. Particularly, when $U = \bigO(n)$, the
universe is called {\em rank space}. In general, for a smaller universe, it may
be possible to achieve better query cost under the same space budget. We consider that $P$ is in general position, i.e., no two points in $P$ have
the same $x$- or $y$-coordinate (datasets not in general position can be
supported by standard tie breaking). When the universe is $[U]^2$, we make the
standard assumption that a machine word has at least $\log_2 U$ bits.

\subsection{Motivation of 2D Range Skyline} \label{sec:intro-motivation}

Skylines have drawn very significant attention (see \cite{BT11, AFGT97,
DGKASK12, FR90, J91, KDKS11, K00, OL81, BCP08, BKS01, CGGL05, KRR02, MPJ07,
PTFS05, SLNX09, SSK09, ST11} and the references therein) from the research
community due to their crucial importance to multi-criteria optimization, which
in turn is vital to numerous applications. In particular, the rectangle of a
range skyline query represents range predicates specified by a user. An
effective index is essential for maximizing the efficiency of these queries in
database systems \cite{KRR02, PTFS05}.

This paper concentrates on 2D data for several reasons. First, {\em planar range skyline reporting} (i.e., our problem) is a classic topic that has been extensively studied in theory  \cite{BT11, AFGT97, DGKASK12, FR90, J91, KDKS11, K00, OL81}. However, nearly all the existing results  apply to internal memory (as reviewed in the next subsection), while currently there is little understanding about the characteristics of the problem in I/O environments.

\begin{table*}
	\centering
	\begin{tabular}{c|c|c|c|c|c}
		 & space & query & insertion & deletion & remark \\
		\hline
		top-open in $\real^2$ & $\bigO(n/B)$ & $\bigO(\log_B n + k/B)$ & - & - & optimal \\
		top-open in $U^2$ & $\bigO(n/B)$ & $\bigO(\log\log_B U + k/B)$ & - & - & optimal \\
		top-open in $[\bigO(n)]^2$ & $\bigO(n/B)$ & $\bigO(1 + k/B)$ & - & - & optimal \\
		\hline
		anti-dominance in $\real^2$ & $\bigO(n/B)$ & $\Omega((n/B)^\eps + k/B)$ & - & - & lower bound (indexability)\\
		4-sided in $\real^2$ & $\bigO(n/B)$ & $\bigO((n/B)^\eps + k/B)$ & - & - & optimal (indexability) \\
		\hline
    top-open in $\real^2$ & $\bigO(n/B)$ & $\bigO(\log_{2B^\eps} n +
    k/B^{1-\eps})$ & $\bigO(\log_{2B^\eps} n)$ & $\bigO(\log_{2B^\eps} n)$ &
    for any constant $\eps \in [0, 1]$ \\
		\hline
		4-sided in $\real^2$ & $\bigO(n/B)$ & $\bigO((n/B)^\eps + k/B)$ & $\bigO(\log (n/B))$ & $\bigO(\log (n/B))$ & update cost is amortized
	\end{tabular}
	\figcapup
	\caption{Summary of our range skyline results (all complexities are in the worst case by default).} \label{tab:intro-results}
	\figcapdown
\end{table*}

The second, more practical, reason is that many skyline applications are {\em inherently} 2D. In fact, the special importance of 2D arises from the fact that one often faces the situation of having to strike a balance between a pair of naturally contradicting factors. A prominent example is {\em price} vs.\ {\em quality} in product selection. A range skyline query can be used to find the products that are not dominated by others in both aspects, when the price and quality need to fall in specific ranges. Other pairs of naturally contradicting factors include {\em space} vs.\ {\em query time} (in choosing data structures), {\em privacy protection} vs.\ {\em disclosed information} (the perpetual dilemma in privacy preservation \cite{CRL07}), and so on.

The last reason, and maybe the most important, is that clearly range skyline
reporting cannot become easier as the dimensionality increases, whereas even for
two dimensions, we will prove a hardness result showing that the problem
(unfortunately) is already difficult enough to forbid sub-polynomial
query cost under the linear space budget! In other words, the ``easiest''
dimensionality of 2 is not so easy after all, which also points to the absence
of query-efficient structures in any higher dimension
when only linear space is permitted.

\subsection{Previous Results} \label{sec:intro-related}

\extraspacing {\bf Range Skyline in Internal Memory.} We first review the
existing results when the dataset $P$ fits in main memory. Early research
focused on dominance and contour queries, both of which can be solved in
$\bigO(\log n + k)$ time using a structure of $\bigO(n)$ size, where $k$ is the
number of points reported \cite{AFGT97, FR90, J91, K00, OL81}. Brodal and
Tsakalidis \cite{BT11} were the first to discover an optimal dynamic structure
for top-open queries, which capture both dominance and contour queries as
special cases. Their structure occupies $\bigO(n)$ space, answers queries in
$\bigO(\log n + k)$ time, and supports updates in $\bigO(\log n)$ time. The
above structures  belong to the {\em pointer machine} model.  Utilizing features
of the RAM model, Brodal and Tsakalidis \cite{BT11} also presented an
alternative structure in universe $[U]^2$, which uses $\bigO(n)$ space, answers
queries in $\bigO(\frac{\log n}{\log\log n} + k)$ time, and can be updated in
$\bigO(\frac{\log n}{\log\log n})$ time. In RAM, the static top-open problem can
be easily settled using an RMQ ({\em range minimum queries}) structure (see,
e.g., \cite{YA10}), which occupies $\bigO(n)$ space and answers queries in
$\bigO(1 + k)$ time.

For general range skyline queries (i.e., 4-sided), all the known structures demand super-linear space. Specifically, Brodal and Tsakalidis \cite{BT11} gave a pointer-machine structure of $\bigO(n \log n)$ size, $\bigO(\log^2 n + k)$ query time, and $\bigO(\log^2 n)$ update time. Kalavagattu et al.\ \cite{KDKS11} designed a static RAM-structure that occupies $O(n \log n)$ space and achieves query time $\bigO(\log n + k)$. In rank space, Das et al.\ \cite{DGKASK12} proposed a static RAM-structure with $\bigO(n \fr{\log n}{\log\log n})$ space and $\bigO(\fr{\log n}{\log\log n} + k)$ query time.

The above results also hold directly in external memory, but they are far from being satisfactory. In particular, all of them incur $\Omega(k)$ I/Os to report $k$ points. An I/O-efficient structure ought to achieve $\bigO(k/B)$ I/Os for this purpose.

\extraspacing {\bf Range Skyline in External Memory.} In contrast to internal
memory where there exist a large number of results, range skyline queries have
not been well studied in external memory. As a naive solution, we can first
scan the entire point set $P$ to eliminate the points falling outside the query
rectangle $Q$, and then find the skyline of the remaining points by the fastest
skyline algorithm \cite{ST11} on non-preprocessed input sets. This expensive
solution can incur $\bigO((n/B) \log_{M/B} (n/B))$ I/Os.

Papadias et al.\ \cite{PTFS05} described a branch-and-bound algorithm when the
dataset is indexed by an R-tree \cite{G84}. The algorithm is heuristic and
cannot guarantee better worst case query I/Os than the naive solution mentioned
earlier. Different approaches have been proposed for skyline maintenance in
external memory under various assumptions on the updates
\cite{TP06,WAEA07,PTFS05,HLOT06}. The performance of those methods, however,
was again evaluated only experimentally on certain ``representative'' datasets.
No I/O-efficient structure exists for answering range skyline queries even in
sublinear I/Os under arbitrary updates.

\extraspacing {\bf Priority Queues with Attrition (PQAs).} Let $S$ be a set of elements drawn from an ordered domain, and let $\min(S)$ be the smallest element in $S$.
A PQA on $S$ is a data structure that supports the following operations:
\begin{itemize}
	\item \textsc{FindMin}: Return $\min(S)$. \vspace{-2mm}

	\item \textsc{DeleteMin}: Remove and return $\min(S)$. \vspace{-2mm}

	\item \textsc{InsertAndAttrite}: Add a new element $e$ to $S$ and remove from $S$ all the elements at least $e$.
	After the operation, the new content is $S'=\{e' \in S \mid e' < e\} \cup \{e\}$.
    The elements $\{e' \in S \mid e' \ge e\}$ are {\em attrited}.
\end{itemize}
In internal memory, Sundar \cite{S89} described how to implement a PQA that supports all operations in $\bigO(1)$ worst case time, and occupies $\bigO(n-m)$ space after $n$ \textsc{InsertAndAttrite} and $m$ \textsc{DeleteMin} operations.

\subsection{Our Results} \label{sec:intro-ours}%

This paper presents external memory structures for solving the planar range skyline reporting problem using only linear space. At the core of one of these structures is a new PQA that supports the extra functionality of catenation. This PQA is a non-trivial extension of Sundar's version \cite{S89}. It can be implemented I/O-efficiently, and is of independent interest due to its fundamental nature. Next, we provide an overview of our results.

\extraspacing {\bf Static Range Skyline.} When $P$ is static, we describe several linear-size structures with the optimal query cost. Our structures also separate the hard variants of the problem from the easy ones.

For top-open queries, we present a structure that answers queries in optimal $\bigO(\log_B n + k/B)$ I/Os (Theorem~\ref{thm:topopen-main}) when the universe is $\real^2$. To obtain the result, we give an elegant reduction of the problem to {\em segment intersection}, which can be settled by a {\em partially persistent B-tree} (PPB-tree) \cite{BGOSW96}. Furthermore, we show that this PPB-tree is (what we call) {\em sort-aware build-efficient} (SABE), namely, it can be constructed in linear I/Os, provided that $P$ is already sorted by $x$-coordinate (Theorem~\ref{thm:topopen-main}). The construction algorithm exploits several intrinsic properties of top-open queries, whereas none of the known approaches \cite{A03, BSW97, V08} for bulkloading a PPB-tree is SABE.

The above structure is {\em indivisible}, namely, it treats each coordinate as an atom by always storing it using an entire word. As the second step, we improve the top-open query overhead beyond the logarithmic bound when the data universe is small. Specifically, when the universe is $[U]^2$ where $U$ is an integer, we give a {\em divisible}  structure with optimal $\bigO(\log \log_B U + k/B)$ query I/Os (Corollary~\ref{crl:div-rankmain}). In the rank space, we further reduce the query cost again optimally to $\bigO(1 + k/B)$ (Theorem~\ref{thm:div-rankmain}).

Clearly, top-open queries are equivalent to right-open queries by symmetry, and capture dominance and contour queries as special cases, so the results aforementioned are applicable to those variants immediately.

Unfortunately, fast query cost with linear space is impossible for the remaining variants under the well-known {\em indexability model} of \cite{HKMPS02} (all the structures in this paper belong to this model). Specifically, for anti-dominance queries, we establish a lower bound showing that every linear-size structure must incur $\Omega((n/B)^\epsilon + k/B)$ I/Os in the worst case (Theorem~\ref{thm:4sided-lower2}), where $\eps > 0$ can be an arbitrarily small constant. Furthermore, we prove that this is tight, by giving a structure to answer a 4-sided query in $\bigO((n/B)^\epsilon + k/B)$ I/Os (Theorem~\ref{thm:4sided-main}). Since 4-sided is more general than anti-dominance, these matching lower and upper bounds imply that they, as well as left- and bottom-open queries, have exactly the same difficulty.

The above 4-sided results also reveal a somewhat unexpected fact: planar range skyline reporting has precisely the same hardness as {\em planar range reporting} (where, given an axis-parallel rectangle $Q$, we want to find all the points in $P \intr Q$, instead of just the maxima; see \cite{ASV99, HKMPS02} for the matching  lower and upper bounds on planar range reporting). In other words, the extra skyline requirement does not alter the difficulty at all.

\extraspacing {\bf Dynamic Range Skyline.} The aforementioned static structures
cannot be updated efficiently when insertions and deletions occur in $P$. For
top-open queries, we provide an alternative structure with fast worst case
update overhead, at a minor expense of query efficiency. Specifically, our
structure occupies linear space, is SABE, answers queries in
$\bigO(\log_{2B^\epsilon} (n/B) + k/B^{1-\epsilon})$ I/Os, and supports updates in $\bigO(\log_{2B^\epsilon} (n/B))$ I/Os, where $\eps$ can
be any parameter satisfying $0\leq \epsilon \leq 1$ (Theorem~\ref{thm:3sided}).
Note that setting $\eps = 0$ gives a structure with query cost $\bigO(\log(n/B)
+ k/B)$ and update cost $\bigO(\log(n/B))$.

The combination of this structure and our (static) 4-sided structure leads to a
dynamic 4-sided structure that uses linear space, answers queries optimally in
$\bigO((n/B)^\epsilon + k/B)$ I/Os, and supports updates in
$\bigO(\log(n/B))$ I/Os amortized (Theorem~\ref{thm:4sided-main}).
Table~\ref{tab:intro-results} summarizes our structures.

\extraspacing {\bf Catenable Priority Queues with Attrition.} A central
ingredient of our dynamic structures is a new PQA that is more
powerful than the traditional version of Sundar \cite{S89}. Specifically, besides  \textsc{FindMin},
\textsc{DeleteMin} and \textsc{InsertAndAttrite} (already reviewed in
Section~\ref{sec:intro-related}), it also supports:
\begin{itemize}
  \item \textsc{CatenateAndAttrite}: Given two PQAs on sets $S_1$ and $S_2$
  respectively, the operation returns a single PQA on $S = \{e \in S_1 \mid e <
  \min(S_2)\} \cup S_2$. In other words, the elements in $\{e \in S_1 \mid e
  \geq \min(S_2)\}$ are attrited.
\end{itemize}
We are not aware of any previous work that addressed the above operation,
which turns out to be rather challenging even in internal
memory.

Our structure, named \emph{I/O-efficient catenable priority queue with
attrition} (I/O-CPQA), supports all operations in $\bigO(1)$ worst case and
$\bigO(1/B)$ amortized I/Os (the amortized bound requires that a constant
number of blocks be pinned in main memory, which is a standard and compulsory
assumption to achieve $\smallO(1)$ amortized update cost of most, if not all,
known structures, e.g., the linked list). The space cost is $\bigO((n-m)/B)$ after
$n$ \textsc{InsertAndAttrite} and \textsc{CatenateAndAttrite} operations, and
after $m$ \textsc{DeleteMin} operations.

\confcmt{All the missing proofs of theorems, lemmata and corollaries can be found in the full version.}
\section{SABE Top-Open Structure} \label{sec:topopen}

In this section, we describe a structure of linear size to answer a top-open
query in $\bigO(\log_B n + k/B)$ I/Os. The structure is SABE, namely, it can be
constructed in linear I/Os provided that the input set $P$ is sorted by
$x$-coordinate.

\subsection{Reduction to Segment Intersection} \label{sec:topopen-reduce}

We first describe a simple structure by converting top-open range skyline
reporting to the {\em segment intersection problem}: the input is a set $S$ of
horizontal segments in $\real^2$; given a vertical segment $q$, a query reports
all the segments of $S$ intersecting $q$.

Given a point $p$ in $P$, denote by $\mathit{leftdom}(p)$ the leftmost
point among all the points in $P$ dominating $p$. If such a point does not
exist, $\mathit{leftdom}(p) =$ \nil. We convert $p$ to a horizontal segment
$\sigma(p)$ as follows. Let $q = \mathit{leftdom}(p)$. If $q =$ \nil, then
$\sigma(p) = [x_p, \infty[ \times y_p$; otherwise,
$\sigma(p) = [x_p, x_q[ \times y_p$. Define $\Sigma(P) = \{\sigma(p) \mid p \in
P\}$, i.e., the set of segments converted from the points of $P$. See
Figure~\ref{fig:topopen-reduce}a for an example.

\begin{figure}[!h]
	\centering
	\begin{tabular}{cc}
    \yufeigraphics{height=30mm}{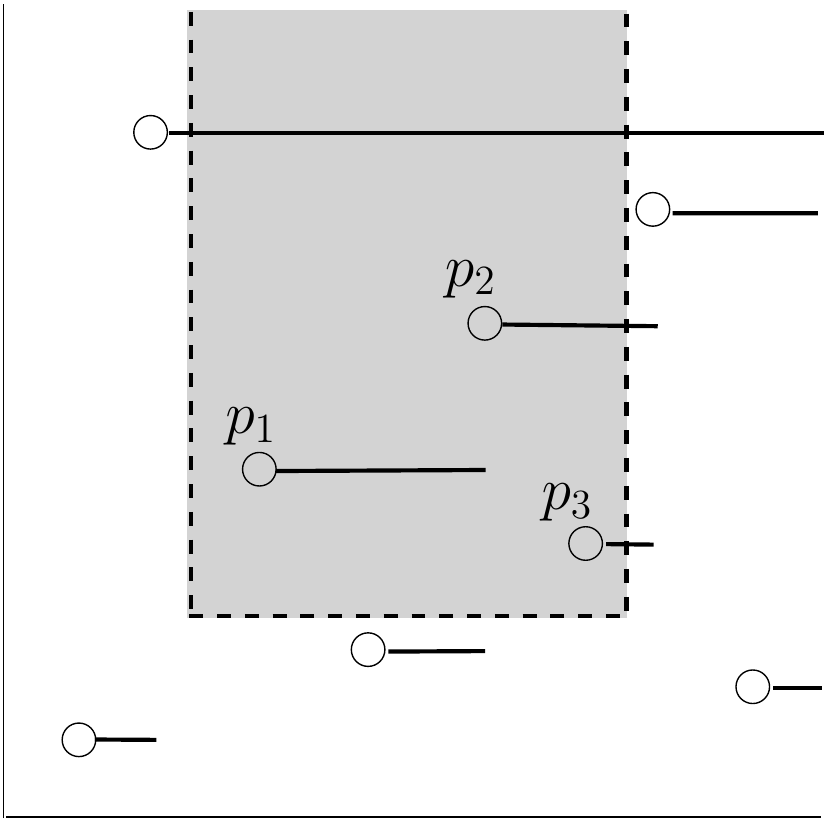}&
    \hspace{3mm}\yufeigraphics{height=30mm}{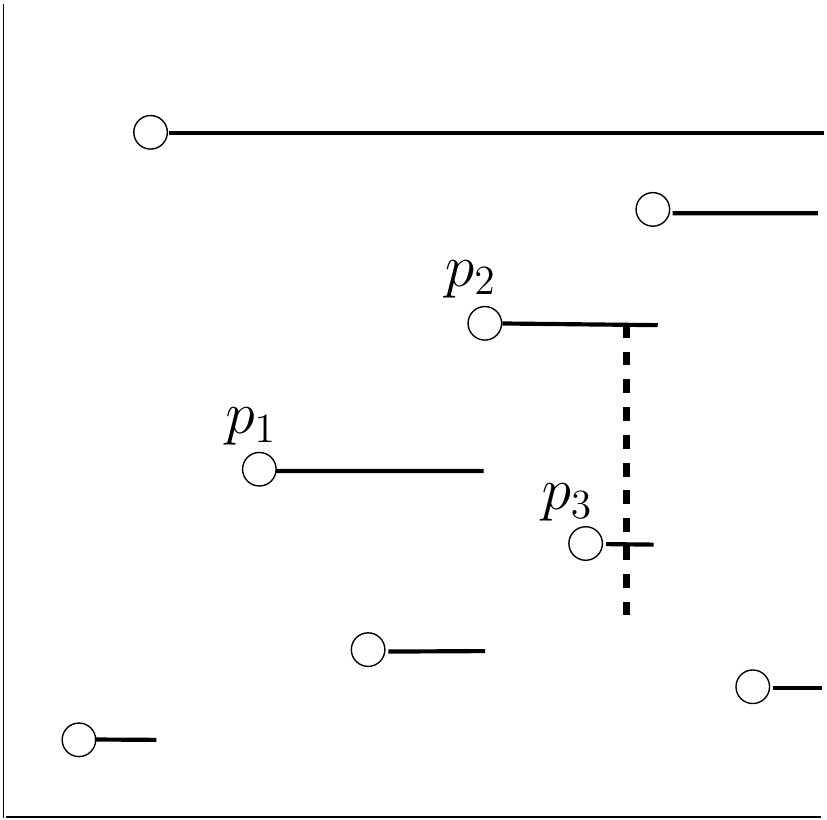} \\
		(a) Data conversion &
		\hspace{3mm} (b) Converted query
	\end{tabular}
	\figcapup
	\caption{Reduction.} \label{fig:topopen-reduce}
	\figcapdown
\end{figure}

Now, consider a top-open query with rectangle $Q = [\alpha_1, \alpha_2] \times
[\beta, \infty[$. We answer it by performing segment intersection
on~$\Sigma(P)$. First, obtain $\beta'$ as the highest $y$-coordinate of the
points in $P \intr Q$. Then, report all segments in $\Sigma(P)$ that intersect
the vertical segment $\alpha_2 \times [\beta, \beta']$. An example is shown in
Figure~\ref{fig:topopen-reduce}b. \confcmt{A proof of the correctness of the
algorithm can be found in the full version.}

\begin{fullenv}
\begin{lemma} \label{lmm:topopen-correct}
	The query algorithm is correct.
\end{lemma}

\begin{proof}
  Consider any point $p \in P$ and a top-open query with $Q = [\alpha_1,
  \alpha_2] \times [\beta, \infty[$. We show that our algorithm reports $p$ if
  and only if $p$ satisfies the query.

  \vspace{2mm} {\em If direction:} As $p$ satisfies the query, we know that $p
  \in Q$, $y_p \le \beta'$, and $q = \leftdom(p) \notin Q$. The last fact
  suggests that $x_q > \alpha_2$ (if $q =$ \nil, define $x_q =
  \infty$). Hence, $\sigma(p) = [x_p, x_q[ \times y_p$ intersects the vertical
  segment $\alpha_2 \times [\beta, \beta']$, and thus, will be reported by our
  algorithm.

  \vspace{2mm} {\em Only-if direction:} Let $p$ be a point found by our
  algorithm, i.e., $\sigma(p) = [x_p, x_q[ \times y_p$ intersects $\alpha_2
  \times [\beta, \beta']$, where $q = \leftdom(p)$ (if $q$ does not exist, $x_q
  = \infty$). It follows that $x_p \le \alpha_2 < x_q$ and $\beta \le y_p \le
  \beta'$.

  Next, we prove $\alpha_1 \le x_p$. Recall that $\beta'$ is the $y$-coordinate
  of the highest point $p'$ among all the points in $P \intr Q$. If $p = p'$,
  then $\alpha_1 \le x_p$ clearly holds. Otherwise, we know $y_p \leq y_{p'}$,
  which implies that $x_p > x_{p'}$. This is because if $x_p \leq x_{p'}$, then
  $p'$ dominates $p$, which (because $x_{p'} \le \alpha_2 < x_q$) contradicts
  the definition of $q$. Now, $x_p \ge \alpha_1$ follows from $x_{p'} \ge
  \alpha_1$.

  So far we have shown that $p$ is covered by $Q$. It remains to prove that $p$
  is not dominated by any point in $P \intr Q$. This is true because $\alpha_2
  < x_q$ suggests that the leftmost point in $P$ dominating $p$ must be outside
  $Q$.
\end{proof}
\end{fullenv}

We can find $\beta'$ in $\bigO(\log_B n)$ I/Os with a {\em range-max query} on
a B-tree indexing the $x$-coordinates in $P$. For retrieving the segments
intersecting $\alpha_2 \times [\beta, \beta']$, we store $\Sigma(P)$ in a
partially persistent B-tree (PPB-tree) \cite{BGOSW96}. As $\Sigma(P)$ has $n$
segments, the PPB-tree occupies $\bigO(n/B)$ space and answers a segment
intersection query in $\bigO(\log_B n + k/B)$ I/Os. We thus have obtained a
linear-size top-open structure with $\bigO(\log_B n + k/B)$ query I/Os.

More effort, however, is needed to make the structure SABE. In particular, two challenges are to be overcome. First, we must generate $\Sigma(P)$ in linear
I/Os. Second, the PPB-tree on $\Sigma(P)$ must be built with asymptotically the same cost (note that the range-max B-tree is already SABE). We will tackle these challenges in the rest of this section.

\subsection{Computing {\large \boldmath$\Sigma(P)$}} \label{sec:topopen-genseg}

$\Sigma(P)$ is not an arbitrary set of segments. We observe:

\begin{lemma} \label{lmm:topopen-properties}
	$\Sigma(P)$ has the following properties:
	\begin{itemize}
    \item {\bf (Nesting)} for any two segments $s_1$ and $s_2$ in $\Sigma(P)$,
    their $x$-intervals are either disjoint, or such that one $x$-interval
    contains the other.

    \item {\bf (Monotonic)} let $\ell$ be any vertical line, and $S(\ell)$ the
    set of segments in $\Sigma(P)$ intersected by $\ell$. If we sort the
    segments of $S(\ell)$ in ascending order of their $y$-coordinates, the lengths of
    their $x$-intervals are non-decreasing.
	\end{itemize}
\end{lemma}
\begin{fullenv}
\begin{proof}
  \extraspacing {\em Nesting:} Let $p_1$ and $p_2$ be the points such that $s_1
  = \sigma(p_1)$ and $s_2 = \sigma(p_2)$. Assume without loss of generality
  that $x_{p_1} < x_{p_2}$. Consider first the case $y_{p_1} < y_{p_2}$. In
  this scenario, the $x$-interval of $s_1$ must terminate before $x_{p_2}$
  because $p_2$ dominates $p_1$. In other words, $s_1$ and $s_2$ have disjoint
  $x$-intervals.

  We now discuss the case $y_{p_1} > y_{p_2}$. If $\leftdom(p_1)$ does not
  exist, the $x$-interval of $s_1$ is $[x_{p_1}, \infty[$, which clearly
  encloses that of $s_2$. Consider, instead, that $\leftdom(p_1)$ exists.  If
  $\leftdom(p_1)$ has $x$-coordinate smaller than $x_{p_2}$, then $s_1$ and $s_2$
  have disjoint $x$-intervals.  Otherwise, $\leftdom(p_1)$ also dominates
  $p_2$, implying that the $x$-interval of $s_1$ contains that of $s_2$.

  \vspace{2mm} {\em Monotonic:} Let $\ell$ intersect the $x$-axis at $\alpha$.
  Consider the contour query with rectangle $Q = ]-\infty, \alpha] \times
  ]-\infty, \infty[$, which is a special top-open query. By
  Lemma~\ref{lmm:topopen-correct}, the left endpoints of the segments in
  $S(\ell)$ constitute the skyline of $P \intr Q$. Therefore, if we enumerate
  the segments of $S(\ell)$ in ascending order of $y$-coordinates, their left
  endpoints' $x$-coordinates decrease continuously. It thus follows from the
  nesting property that their $x$-intervals have increasing lengths.
\end{proof}
\end{fullenv}

We are ready to present our algorithm for computing $\Sigma(P)$, after $P$ has been sorted by $x$-coordinates. Conceptually, we sweep a vertical line $\ell$ from $x = -\infty$ to $\infty$. At any time, the algorithm (essentially) stores the set $S(\ell)$ of segments in a stack, which are en-stacked in descending order of $y$-coordinates (i.e., the segment that tops the stack has the lowest y-coordinate). Whenever a segment is popped out of the stack, its right endpoint is decided, and the segment is output. In general, the segments of $\Sigma(P)$ are output in non-descending order of their right endpoints' $x$-coordinates.

Specifically, the algorithm starts by pushing the leftmost
point of $P$ onto the stack. Iteratively, let $p$ be the next point fetched
from $P$, and $q$ the point currently at the top of the stack. If $y_q < y_p$, we know that $p = \leftdom(q)$. Hence, the algorithm pops
$q$ off the stack, and outputs segment $\sigma(q) = [x_q, x_p[ \times y_q$.
Then, letting $q$ be the point that tops the stack currently, the algorithm
checks again whether $y_q < y_p$, and if so, repeats the above steps. This
continues until either the stack is empty or $y_q > y_p$. In either case, the
iteration finishes by pushing $p$ onto the stack. It is clear that the algorithm generates $\Sigma(P)$ in $\bigO(n/B)$ I/Os.

\subsection{Constructing the PPB-tree} \label{sec:topopen-pers}

Remember that we need a PPB-tree $T$ on $\Sigma(P)$. The known algorithms for
PPB-tree construction require super-linear I/Os even after sorting \cite{A03,
BGOSW96, BSW97, V08}. Next, we show that the two properties of $\Sigma(P)$ in
Lemma~\ref{lmm:topopen-properties} allow building~$T$ in linear I/Os. Let us
number the leaf level as {\em level 0}. In general, the parent of a level-$i$
($i \ge 0$) node is at level $i+1$. We will build $T$ in a bottom-up manner,
i.e., starting from the leaf level, then level $1$, and so on.

\extraspacing{\bf Leaf Level.} To create the leaf nodes, we need to first sort
the left and right endpoints of the segments in $\Sigma(P)$ together by
$x$-coordinate. This can be done in $\bigO(n/B)$ I/Os as follows. First, $P$,
which is sorted by $x$-coordinates, gives a sorted list of the left endpoints.
On the other hand, our algorithm of the previous subsection generates
$\Sigma(P)$ in non-descending order of the right endpoints' $x$-coordinates
(breaking ties by favoring lower points). By merging the two lists, we obtain
the desired sorted list of left and right endpoints combined.

Let us briefly review the algorithm proposed in \cite{BGOSW96} to build a
PPB-tree. The algorithm conceptually moves a vertical line $\ell$ from $x =
-\infty$ to $\infty$. At any moment, it maintains a B-tree $T(\ell)$ on the
$y$-coordinates of the segments in $S(\ell)$. We call $T(\ell)$ a {\em
snapshot B-tree}. To do so, whenever $\ell$ hits the left (resp.\ right) endpoint of a
segment $s$, it inserts (resp.\ deletes) the $y$-coordinate of $s$ in $T(\ell)$. The
PPB-tree can be regarded as a space-efficient union of all the snapshot
B-trees. The algorithm incurs $\bigO(n \log_B n)$ I/Os because (i) there are
$2n$ updates, and (ii) for each update, $\bigO(\log_B n)$ I/Os are needed to
locate the leaf node affected.

When $\Sigma(P)$ is nesting and monotonic, the construction
can be significantly accelerated. A crucial observation is that any update to
$S(\ell)$ happens only {\em at the bottom} of $\ell$. Specifically, whenever
$\ell$ hits the left/right endpoint of a segment $s \in \Sigma(P)$, $s$ must be
the lowest segment in $S(\ell)$. This implies that the leaf node of $T(\ell)$
to be altered must be the leftmost\footnote{We adopt the convention that the leaf elements of a B-tree are ordered from left to right in ascending order.} one in $T(\ell)$. Hence, we can find this leaf
without any I/Os by buffering it in memory, in contrast to the $\bigO(\log_B n)$
cost originally needed.

The other details are standard, and are sketched below assuming the knowledge
of the classic algorithm in \cite{BGOSW96}. Whenever the leftmost leaf $u$ of $T(\ell)$ is
full, we version copy it to $u'$, and possibly perform a split or merge, if
$u'$ strong-version overflows or underflows, respectively\footnote{Version copy, strong-version overflow and strong-version underflow are
concepts from the terminology of \cite{BGOSW96}.}. A version
copy, split, and merge can all be handled in $\bigO(1)$ I/Os, and can happen
only $\bigO(n/B)$ times.  Therefore, the cost of building the leaf level is
$\bigO(n/B)$.

\extraspacing{\bf Internal Levels.} The level-$1$ nodes can be built by exactly
the same algorithm, but on a different set of segments $\Sigma_1$ which are
generated from the leaf nodes of the PPB-tree. To explain, let us first review
an intuitive way \cite{BS96} to visualize a node in a PPB-tree. A node $u$ can
be viewed as a rectangle $r(u) = [x_1, x_2[ \times [y_1, y_2[$ in $\real^2$, where $x_1$ (resp.\ $x_2$) is the position of $\ell$ when $u$ is created (resp.\ version copied), and $[y_1, y_2[$ represents the $y$-range of $u$ in all the snapshot B-trees where $u$ belongs. See Figure~\ref{fig:topopen-mvb}.

\begin{figure}[!h]
	\centering
    \yufeigraphics{height=35mm}{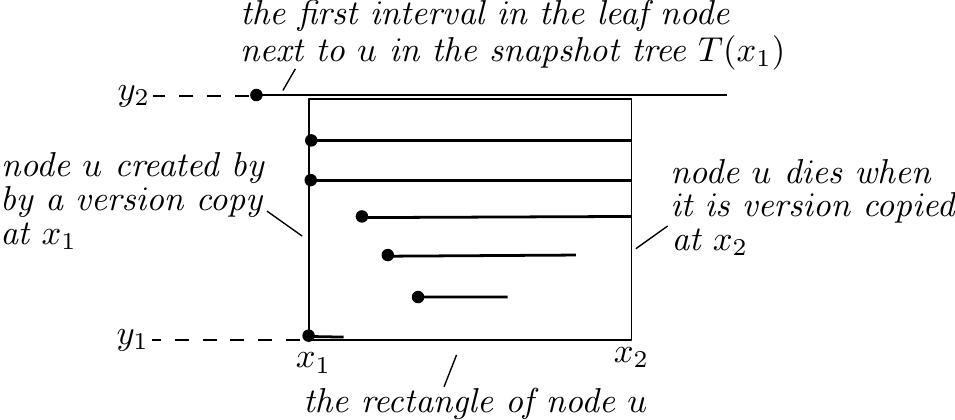}
	\figcapup
	\caption{A node in a PPB-tree.}
	\label{fig:topopen-mvb}
	\figcapdown
\end{figure}

For each leaf node $u$ (already created), we add the bottom edge of $r(u)$, namely $[x_1, x_2[ \times
y_1$, into $\Sigma_1$. The next lemma points out a crucial fact.

\begin{lemma} \label{lmm:topopen-level1}
	$\Sigma_1$ is both nesting and monotonic.
\end{lemma}
\begin{fullenv}
\begin{proof}
  We prove the lemma by induction on the position of $\ell$. For this purpose,
  care must be taken to interpret the rectangles of the nodes currently in
  $T(\ell)$.  As these nodes are still ``alive'' (i.e., they have not been
  version copied yet), the right edges of their rectangles rest on $\ell$, and
  move rightwards along with $\ell$. Let set $\Sigma_1(\ell)$ include the
  bottom edges of the rectangles of all level-1 nodes already spawned so far,
  counting also the ones in $T(\ell)$. When we finish building all the level-1
  nodes, $\Sigma_1(\ell)$ becomes the final $\Sigma_1$. We will show that
  $\Sigma_1(\ell)$ is nesting and monotonic at all times. This is obviously
  true when $\ell$ is at $x = -\infty$.

  Now, suppose that $\Sigma_1(\ell)$ is currently nesting and monotonic. We
  will prove that it remains so after the next update on $T(\ell)$. This is
  trivial if the update does not cause any version copy, i.e., the first leaf
  node $u$ of $T(\ell)$ is not full yet. Consider instead that $u$ is version
  copied to $u'$ when $\ell$ is at $x = \alpha$. At this point, $r(u)$ is
  finalized. Because $r(u)$ is the lowest among the rectangles of the nodes in
  $T(\ell)$, its finalization cannot affect the nesting and monotonicity of
  $\Sigma_1(\ell)$. The version copy also creates $r(u')$. Note that the
  x-intervals of $r(u)$ and $r(u')$ are disjoint, because the former does not
  include $\alpha$, but the latter does. Furthermore, $r(u')$ has the same
  $y$-interval as $r(u)$, and a zero-length $x$-interval $[\alpha, \alpha]$.
  Therefore, if no split/merge follows, $\Sigma_1(\ell)$ is still nesting and
  monotonic.

  Next, consider that $u'$ is split into $u'_1$ and $u'_2$. In this case,
  $r(u')$ disappears from $\Sigma_1(\ell)$, and is replaced by $r(u'_1)$ and
  $r(u'_2)$, which are the bottom two among the rectangles of the nodes in
  $T(\ell)$. Furthermore, both $r(u'_1)$ and $r(u'_2)$ have zero-length
  $x$-intervals. So $\Sigma_1(\ell)$ is still nesting and monotonic.

  It remains to discuss the case where $u'$ needs to merge with its sibling $v$
  in $T(\ell)$. When this happens, the algorithm first version copies $v$ to
  $v'$, which finalizes $r(v)$. The $x$-interval of $r(v)$ must contain that of
  $r(u)$, which is consistent with nesting and monotonicity because $r(v)$ is
  above $r(u)$. The merge of $u'$ and $v'$ creates a node $z$, such that $r(z)$
  has a zero-length $x$-interval. Note that $r(z)$ is currently the lowest of
  the rectangles of the nodes in $T(\ell)$. So $\Sigma_1(\ell)$ remains nesting
  and monotonic.

  Finally, $z$ may still need to be split one more time, but this case can be
  analyzed in the same way as the split scenario mentioned earlier. We thus
  conclude the proof.
\end{proof}
\end{fullenv}

Our algorithm (for building the leaf nodes) writes the left and right endpoints
of the segments in $\Sigma_1$ in non-descending order of their $x$-coordinates
(breaking ties by favoring lower endpoints). This, together with
Lemma~\ref{lmm:topopen-level1}, permits us to create the level-$1$ nodes using
the same algorithm in $\bigO(n/B^2)$ I/Os (as $|\Sigma_1| = \bigO(n/B)$). We
repeat the above process to construct the nodes of higher levels. The cost
decreases by a factor of $B$ each level up. The overall construction cost is
therefore $\bigO(n/B)$. \confcmt{Leaving the other details to the full version,
we now conclude with the first main result:}

\begin{theorem} \label{thm:topopen-main}
  There is an indivisible linear-size structure on $n$ points in $\real^2$, such
  that top-open range skyline queries can be answered in $\bigO(\log_B n +
  k/B)$ I/Os, where $k$ is the number of reported points. If all points have
  been sorted by $x$-coordinates, the structure can be built in linear I/Os.
  The query cost is optimal (even without assuming indivisibility).
\end{theorem}
\begin{fullenv}
\begin{proof}
  We focus on the query optimality because the rest of the theorem follows from
  our earlier discussion directly.

  The term $k/B$ is clearly indispensable. The term $\bigO(\log_B n)$, on the
  other hand, is also compulsory due to a reduction from predecessor search.
  First, it is well-known (see, e.g., \cite{BKMT97}) that predecessor search
  can be reduced to top-open range reporting (note: {\em not} top-open range
  skyline), such that if a linear-size structure can answer a top-open range
  query in $f(n, B) + \bigO(k/B)$ time, the same structure also solves a
  predecessor query in $f(n, B)$ time. Interestingly, given a predecessor
  query, the converted top-open range query always returns only one point.
  Hence, the query can as well be interpreted as a top-open range skyline
  query. This indicates that the same reduction also works from predecessor
  search to top-open range skyline. Finally, any linear-size structure must
  incur $\Omega(\log_B n)$ I/Os answering a predecessor query in the worst case
  \cite{PT06} (even without the indivisibility assumption). It thus follows that $\Omega(\log_B n)$ also lower bounds the
  cost of a top-open range skyline query.
\end{proof}
\end{fullenv}

\section{Divisible Top-Open Structure} \label{sec:div}

The structure of the previous section obeys the indivisibility assumption. This section
eliminates the assumption, and unleashes the power endowed by bit manipulation.
As we will see, when the universe is small, it admits linear-size structures
with lower query cost.

In Section \ref{sec:div-ray}, we study a different problem called ray-dragging. Then, in Section \ref{sec:div-topopen}, our ray-dragging structure is deployed to develop a ``few-point structure'' for answering top-open queries on a small point set. Finally, in Section \ref{sec:div-final}, we combine our few-point structure with an existing structure \cite{BT11} to obtain the final optimal top-open structure.

\subsection{Ray Dragging} \label{sec:div-ray}

In the {\em ray dragging problem}, the input is a set $S$ of $m$ points in $[U]^2$ where $U \ge m$ is an integer. Given a vertical ray $\rho = \alpha \times [\beta, U]$ where $\alpha, \beta \in [U]$, a ray dragging query reports the first point in $S$ to be hit by $\rho$ when $\rho$ moves left. The rest of the subsection serves as the proof for:

\begin{lemma} \label{lmm:div-ray}
  For $m = (B \log U)^{\bigO(1)}$, we can store $S$ in a structure of size
  $\bigO(1 + m/B)$ that can answer ray dragging queries in $\bigO(1)$ I/Os.
\end{lemma}


\noindent {\bf Minute Structure.} Set $b = B \log_2 U$. We first consider the scenario where $S$ has very few points: $m \le b^{1/3}$. Let us convert $S$ to a set $S'$ of points in an $m \times m$ grid. Specifically, map a point $p \in S$ to $p' \in S'$ such that $x_{p'}$ (resp.\ $y_{p'}$) is the {\em rank} of $x_{p}$ (resp.\ $y_p$) among the $x$- ($y$-) coordinates in $S$.

Given a ray $\rho = \alpha \times [\beta, \infty[$, we instead answer a query in $[m]^2$ using a ray $\rho' = \alpha' \times [\beta', \infty[$, where $\alpha'$ (resp.\ $\beta'$) is the rank of the predecessor of $\alpha$ (resp.\ $\beta$) among the $x$- (resp.\ $y$-) coordinates in $S$. Create a {\em fusion tree} \cite{FW93, LP12} on the $x$- (resp.\ $y$-) coordinates in $S$ so that the predecessor of~$\alpha$ (resp.\ $\beta$) can be found in $\bigO(\log_b m) = \bigO(1)$ I/Os, which is thus also the cost of turning $\rho$ into $\rho'$. The fusion tree uses $\bigO(1 + m/B)$ blocks.

We will ensure that the query with $\rho'$ (in $[m]^2$) returns an id from 1 to
$m$ that uniquely identifies a point $p$ in $S$, if the result is non-empty. To
convert the id into the coordinates of $p$, we store $S$ in an array of
$\bigO(1 + m/B)$ blocks such that any point can be retrieved in one I/O by id.

The benefit of working with $S'$ is that
each coordinate in $[m]^2$ requires fewer bits to represent (than in $[U]^2$),
that is, $\log_2 m$ bits. In particular, we need $3 \log_2 m$ bits in total to
represent a point's $x$-, $y$-coordinates, and id. Since $|S'| = m$, the
storage of the entire $S'$ demands $3m \log m = \bigO(b^{1/3} \log_2 b)$ bits. If $B \geq \log_2 U$, then $b^{1/3} \log_2 b = \bigO((B^2)^{1/3} \log_2 (B^2)) = \bigO(B)$. On the other hand, if $B < \log_2 U$, then $b^{1/3} \log_2 b = \bigO((\log_2^2 U)^{1/3} \log_2 (\log_2^2 U))
= \bigO(\log_2 U)$. In other words, we can always store the entire set $S'$ in 
$O(1)$ blocks. Given a query with $\rho'$, we simply load this block into
memory, and answer the query in memory with no more I/O.

We have completed the description of a structure that uses $\bigO(1 + m/B)$ blocks, and answers queries in constant I/Os when $m \le b^{1/3}$. We refer to it as a {\em minute structure}.

\extraspacing {\bf Proof of Lemma~\ref{lmm:div-ray}.} We store $S$ in a B-tree that indexes the $x$-coordinates of the points in $S$. We set the B-tree's leaf capacity to $B$ and internal fanout to $f = b^{1/3}$. Note that the tree has a constant height.

Given a node $u$ in the tree, define $Y_{max}(u)$ as the highest point whose $x$-coordinate is stored in the subtree of $u$. Now, consider $u$ to be an internal node with child nodes $v_1, ..., v_{f}$. Define $Y^*_{max}(u) = \{Y_{max}(v_i) \mid 1 \le i \le f\}$. We store $Y^*_{max}(u)$ in a minute structure. Also, for each point $p \in Y^*_{max}(u)$, we store an index indicating the child node whose subtree contains the $x$-coordinate of $p$. A child index requires $\log_2 b^{1/3} = \bigO(\log_2 m) = \bigO(\log U)$ bits, which is no more than the length of a coordinate. Hence, we can store the index along with $p$ in the minute structure without increasing its space by more than a constant factor. For a leaf node $z$, define $Y^*_{max}(z)$ to be the set of points whose $x$-coordinates are stored in $z$.

Since there are $\bigO(1 + m/(b^{1/3}B))$ internal nodes and each
minute structure demands $\bigO (1 + b^{1/3}/B)$ space, all the minute structures
occupy $\bigO( (1 + \frac{m}{b^{1/3}B})(\frac{b^{1/3}}{B}+1) )= \bigO(1 + m/B)$ blocks in total. Therefore, the overall structure consumes linear space.

We answer a ray-dragging query with ray $\rho = \alpha \times [\beta, U]$ as follows. First, descend a root-to-leaf path $\pi$ to the leaf node containing the predecessor of $\alpha$ among the $x$-coordinates in~$S$. Let $u$ be the {\em lowest} node on $\pi$ such that $Y^*_{max}(u)$ has a point that can be hit by $\rho$ when $\rho$ moves left. For each node $v \in \pi$, whether $Y^*_{max}(v)$ has such a point can be checked in $\bigO(1)$ I/Os by querying the minute structure over $Y^*_{max}(v)$. Hence, $u$ can be identified in $\bigO(h)$ I/Os where $h$ is the height of the B-tree. If $u$ does not exist, we return an empty result (i.e., $\rho$ does not hit any point no matter how far it moves).

If $u$ exists, let $p$ be the first point in $Y^*_{max}(u)$ hit by $\rho$ when it moves left. Suppose that the $x$-coordinate of $p$ is in the subtree of $v$, where $v$ is a child node of $u$. The query result must be in the subtree of $v$, although it may not necessarily be $p$. To find out, we descend another path from $v$ to a leaf. Specifically, we set $u$ to $v$, and find the first point $p$ in $Y^*_{max}(u)$ ($= Y^*_{max}(v)$) that is hit by $\rho$ when it moves left (notice that $p$ has changed). Now, letting $v$ be the child node of $u$ whose subtree $p$ is from, we repeat the above steps. This continues until $u$ becomes a leaf, in which case the algorithm returns $p$ as the final answer.
The query cost is $\bigO(h) = \bigO(1)$. This completes the proof of Lemma~\ref{lmm:div-ray}. We will refer to the above structure as a {\em ray-drag tree}.

\subsection{Top-Open Structure on Few Points} \label{sec:div-topopen}

Next, we present a structure for answering top-open queries on small $P$, called henceforth the {\em few-point structure}. Remember that $P$ is a set of $n$ points in $[U]^2$ for some integer $U \ge n$, and a query is a rectangle $Q = [\alpha_1, \alpha_2] \times [\beta, U]$ where $\alpha_1, \alpha_2, \beta \in [U]$.

\begin{lemma} \label{lmm:div-skysmall}
  For $n \le (B \log U)^{\bigO(1)}$, we can store $P$ in a structure of $\bigO(1 + n/B)$ space
  that answers top-open range skyline queries  with output size $k$ in $\bigO(1 + k/B)$
  I/Os.
\end{lemma}
\begin{proof}
  Consider a query with $Q = [\alpha_1, \alpha_2] \times [\beta, U]$. Let~$p$
  be the first point hit by the ray $\rho = \alpha_2 \times [\beta, U]$ when
  $\rho$ moves left. If $p$ does not exist or is out of $Q$ (i.e., $x_p <
  \alpha_1$), the top-open query has an empty result. Otherwise, $p$ must be
  the lowest point in the skyline of $P \intr Q$.

  The subsequent discussion focuses on the scenario where $p \in Q$. We index
  $\Sigma(P)$ with a PPB-tree $T$, as in Theorem~\ref{thm:topopen-main}. Recall
  that the top-open query can be solved by retrieving the set $S$ of segments
  in $\Sigma(P)$ intersecting the vertical segment $\psi = \alpha_2 \times
  [\beta, \beta']$, where $\beta'$ is the highest $y$-coordinate of the points
  in $P \intr Q$. To do so in $\bigO(1 + k/B)$ I/Os, we utilize the next two
  observations. \confcmt{(see the full version for their proofs)}:
  \begin{observation}
    {\em All segments of $S$ intersect $\psi' = x_p \times [y_p, \beta']$.}
  \end{observation}

  \fullcmt{\underline{\em Proof:} $\sigma(p)$ is the lowest among the segments of $\Sigma(P)$
    intersecting $\psi$ (recall that $\sigma(p)$ is the segment in $\Sigma(P)$
    converted from $p$). Hence, a segment of $\Sigma(P)$ intersects $\psi$ if
    and only if it intersects $\alpha_2 \times [y_p, \beta']$. On the other
    hand, a segment of $\Sigma(P)$ intersects $\alpha_2 \times [y_p, \beta']$
    if and only if it intersects $\psi'$. To explain, let $s \neq \sigma(p)$ be
    a segment in $\Sigma(P)$ intersecting $\alpha_2 \times [y_p, \beta']$. As
    $s$ is higher than $\sigma(p)$, the $x$-interval of $s$ must contain that
    of $\sigma(p)$ (due to the nesting and monotonicity properties of
    $\Sigma(P)$), implying that $s$ intersects $\psi'$. Similarly, one can also
    show that if $s$ intersects $\psi'$, it also intersects $\alpha_2 \times
    [y_p, \beta']$.
  }

  \begin{observation}
    {\em Let $T(\ell)$ be the snapshot B-tree in $T$ when $\ell$ is at the
    position $x = x_p$. Once we have obtained the leaf node in $T(\ell)$
    containing $y_p$, we can retrieve $S$ in $\bigO(1 + k/B)$ I/Os without
    knowing the value of $\beta'$.}
  \end{observation}

  \fullcmt{
    \underline{\em Proof:} Each leaf node in $T(\ell)$ has a sibling pointer to
    its succeeding leaf node\footnote{Due to the nesting and monotonicity
    properties, every leaf node $u$ in the PPB-tree $T$ needs only one sibling
  pointer during the entire period when $u$ is alive.}. Hence, starting from
  the leaf node storing $y_p$, we can visit the leaves of $T(\ell)$ in
  ascending order of the $y$-coordinates they contain. The effect is to report
  in the bottom-up order the segments of $\Sigma(P)$ that intersect $x_p \times
  [y_p, U]$. By the nesting and monotonicity properties, the left endpoint of a
  segment reported latter has a smaller $x$-coordinate. We stop as soon as
  reaching a segment whose left endpoint falls out of $Q$. The cost is $\bigO(1
  + k/B)$ because $\Omega(B)$ segments are reported in each accessed leaf,
  except possibly the last one. \vspace{2mm}
  }

  We now elaborate on the structure of Lemma~\ref{lmm:div-skysmall}.
  Besides~$T$, also create a structure of Lemma~\ref{lmm:div-ray} on $P$.
  Moreover, for every point $p \in P$, keep a pointer to the leaf node of $T$
  that (i) is in the snapshot B-tree $T(\ell)$ when $\ell$ is at $x = x_p$, and
  (ii) contains $y_p$. Call the leaf node the {\em host leaf} of $p$. Store the
  pointers in an array of size $n$ to permit retrieving the pointer of any
  point in one I/O.

  The query algorithm should have become straightforward from the above two
  observations. We first find in $O(1)$ I/Os the first point $p$ hit by $\rho$
  when $\rho$ moves left. Then, using~$p$, we jump to the host leaf of $p$.
  Next, by Observation 2, we retrieve $S$ in $O(1 + k/B)$ I/Os. The total query
  cost is $\bigO(1 + k/B)$.
\end{proof}

\subsection{Final Top-Open Structure} \label{sec:div-final}

We are ready to describe our top-open structure that achieves sub-logarithmic
query I/Os for arbitrary $n$. For this purpose, we externalize an
internal-memory structure of \cite{BT11}. The structure of \cite{BT11}, however,
has logarithmic query overhead, which we improve with new ideas based on the
few-point structure in Lemma~\ref{lmm:div-skysmall}. \confcmt{Delegating the
details to the full version, we now state our main results in rank space and
universe $[U]^2$:}

\begin{theorem} \label{thm:div-rankmain}
  There is a linear-size structure on $n$ points in rank space such that
  top-open range skyline queries can be answered optimally in $\bigO(1 + k/B)$
  I/Os, where $k$ is the number of reported points.
\end{theorem}

\begin{fullenv}
\noindent {\bf Structure.} Let $U = \bigO(n)$ be the length of each dimension. We assume, without loss of generality, that 
$\lambda = B \log^2 U$ is an integer. Divide the $x$-dimension
of $[U]^2$ into $\lceil U/\lambda \rceil$ consecutive intervals of length $\lambda$ each, except possibly the last interval.
Call each interval a {\em chunk}. Assign each point $p \in P$ to the unique
chunk covering $x_p$. Note that some chunks may be empty.

Create a complete binary search tree $\T$ on the chunks. Let $u$ be a node of
$\T$. We say that a point $p$ is ``in the subtree of $u$'' if it is assigned to
a chunk in the subtree of $u$. Denote by $P(u)$ the set of points in the
subtree of $u$. Define $\high(u)$ as the set of $B$ highest points in the
skyline of $P(u)$; if the skyline of $P(u)$ has less than $B$ points,
$\high(u)$ includes all of them. Furthermore, if $|\high(u)| = B$, let
$\highend(u)$ be the lowest point in $\high(u)$; otherwise, $\highend(u) =$
\nil. We store $\high(u)$ along with $u$.

Let $u$ be any internal node such that $p = \highend(u)$ is not \nil. Denote by $\pi(u)$ the path from the leaf (a.k.a.\ chunk) $z$ of $\T$ covering $x_p$ to the child of $u$ that is an ancestor of $z$. Define $\Pi_\gamma(u)$ as the set of right siblings\footnote{If a node is the right child of its parent, it has no right sibling. Similarly, if a node is a left child, it has no left sibling.} of the nodes in $\pi(u)$. Let $\MAX(u)$ be the skyline of the point set $\bigcup_{v \in \Pi_\gamma(u)} \high(v).$ We store $\MAX(u)$ along with $u$, and order the points in $\MAX(u)$ by $x$-coordinate (hence, also by $y$-coordinate). In Figure~\ref{fig:div-max}, for example, $MAX(u)$ is the skyline of $\bigcup_{i=1}^4 \high(v_i)$.

The above completes the externalization of the structure in \cite{BT11}. Next,
we describe new mechanisms for achieving query cost $\bigO(1 + k/B)$. First, we
index the points in each chunk $z$ with a few-point structure of
Lemma~\ref{lmm:div-skysmall}. Moreover, for every $z$ and every proper ancestor $u$ of $z$, we store two sets $\lmax(z, u)$ and $\rmax(z, u)$ defined as follows. Let
$\pi(z, u)$ be the path from $z$ to the child of $u$ that is an ancestor of
$z$. Define $\Pi_\ell(z, u)$ as the set of left siblings of the nodes on
$\pi(z, u)$, and conversely, $\Pi_\gamma(z, u)$ the set of right siblings of
those nodes. Then:
\begin{itemize}
  \item $\lmax(z, u)$ is the skyline of $\bigcup_{v \in \Pi_\ell(z, u)}
      \high(v)$
  \item $\rmax(z, u)$ is the skyline of $\bigcup_{v \in \Pi_\gamma(z, u)}
    \high(v)$.
\end{itemize}
For instance, in Figure~\ref{fig:div-max}, $RMAX(z, u)$ is the skyline of $\bigcup_{i=1}^4 \high(v_i)$, whereas $LMAX(z, u)$ is the skyline of $\high(v_5) \cup \high(v_6)$. The points of both $\lmax(z, u)$ and $\rmax(z, u)$ are sorted by $x$-coordinate.

\extraspacing {\bf Space.} Let $h = \bigO(\log U)$ be the height of $\T$. We
analyze first the space consumed by the $\bigO(U/\lambda)$ internal nodes $u$
of $\T$. Clearly, $\high(u)$ fits in $\bigO(1)$ blocks, whereas $\MAX(u)$ occupies
$\bigO(h)$ blocks. All the internal nodes thus demand $\bigO(h \cdot (U /
\lambda)) = \bigO(U/B) = \bigO(n/B)$ blocks in total.

Now, let us focus on the $\bigO(U/\lambda)$ leaf nodes $z$ of $\T$. As each
few-point structure uses linear space, all the few-point structures demand
$\bigO(U/\lambda + n/B) = \bigO(n/B)$ blocks altogether. Regarding $\lmax(z,
u)$, $z$ has at most $h$ proper ancestors $u$, while each $\lmax(z, u)$
requires $\bigO(h)$ blocks. Hence, the $\lmax(z, u)$ of all $z$ and $u$ occupy
$\bigO((U/\lambda) \cdot h^2) = \bigO(n/B)$ blocks in total. The case with
$\rmax(z, u)$ is symmetric. The overall space consumption is therefore linear.

\extraspacing {\bf Query.} We need the following fact:

\begin{lemma} \label{lmm:div-half}
  Given a node $u$ in $\T$ and a value $\beta$, let $P(u, \beta)$ be the set of
  points in $P(u)$ with $y$-coordinates greater than $\beta$. We can report the skyline of
  $P(u, \beta)$ in $\bigO(1 + k/B)$ I/Os where $k$ is the number of points reported.
\end{lemma}

\begin{proof}
  If $u$ is a leaf, find the skyline of $P(u, \beta)$ by issuing a top-open
  query with search rectangle $[-U, U] \times [\beta, U]$ on the few-point
  structure of $u$. The query time is $\bigO(1 + k/B)$ by
  Lemma~\ref{lmm:div-skysmall}.

  The rest of the proof adapts an argument in \cite{BT11} to external memory.
  Given an internal node $u$, we find the skyline of $P(u, \beta)$ as follows.
  Load $\high(u)$ into memory, and report the points therein with
  $y$-coordinates above $\beta$. If there are less than $B$ such points, we
  have found the entire skyline of $P(u, \beta)$.

  Suppose instead that the entire $\high(u)$ is reported. Let $p =
  \highend(u)$. It suffices to consider the points that
	\begin{itemize}
		\item[(i)] are in the subtrees of the nodes in $\Pi_\gamma(u)$, or

		\item[(ii)] share the same chunk as $p$, but are to the right of $p$.
	\end{itemize}
  Any other point of $P(u)$ must be either in $\high(u)$ -- which is already
  found -- or dominated by $p$.

  To find the skyline points in (i), first report the set $S$ of points in
  $\MAX(u)$ whose $y$-coordinates are above $\beta$. Then, we explore the
  subtrees of certain nodes in $\Pi_\gamma(u)$. Specifically, let $v_1, ...,
  v_c$ be the nodes in $\Pi_\gamma(u)$ for some integer $c$. For each $i \in
  [1, c]$, define $S_i = \high(v_i) \intr S$; if $|S_i| < B$,\footnote{This can
  be checked efficiently because the points of $\high(v_i)$ are consecutive in
  $MAX(u)$.} the subtree of $v_i$ can be pruned from further
  consideration\footnote{This means that either $|\high(v_i)| < B$, or
  $\highend(v_i)$ is dominated by a point in $S$. In both cases, we have found
  all the result points from the subtree of $v_i$.}. Otherwise (i.e., $|S_i| =
  B$), we recursively report the skyline of $P(v_i, \beta_i)$, where $\beta_i$
  is the $y$-coordinate of the point just to the right of $\highend(v_i)$ in the
  staircase of $S$; if no such point exists, $\beta_i = \beta$.

  The skyline points in (ii) can be retrieved with a top-open query on the
  few-point structure of the chunk $z$ covering $x_p$, where $z$ can be
  identified in constant I/Os by dividing $x_p$ by $\lambda$. Specifically, if
  $S \neq \emptyset$, define $\beta_0$ to be the $y$-coordinate of the highest
  point in $S$; otherwise, define $\beta_0 = \beta$. The top-open query for $z$
  has rectangle $]x_p, U] \times ]\beta_0, U]$.

  Now we analyze the query cost. If less than $B$ points of $\high(u)$ are
  reported, the algorithm finishes with $\bigO(1)$ I/Os. Otherwise, the scan of
  $\MAX(u)$ takes $\bigO(1 + |S|/B)$ I/Os. If $|S| < B$, we charge the $\bigO(1
  + |S|/B) = \bigO(1)$ cost on the $B$ points in $\high(u)$; otherwise, we
  charge the $\bigO(|S|/B)$ cost on the points of $S$. The top-open query on
  the few-point structure of $z$ requires $\bigO(1 + k'/B)$ I/Os if it returns
  $k'$ points. If $k' < B$, we charge the $\bigO(1 + k'/B) = \bigO(1)$ cost on
  the points of $\high(u)$; otherwise, charge the $\bigO(k'/B)$ I/Os on the
  $k'$ points.

  It remains to discuss the I/Os spent on $v_1, ..., v_c$. For each $i \in [1,
  c]$, if $|S_i| < B$, there is no cost on $v_i$. Otherwise, we charge on the
  $B$ points of $S_i$ the $\bigO(1)$ I/Os spent on reading $\high(v_i)$ before
  recursively reporting the skyline of $P(v_i, \beta_i)$. The rest of the I/Os
  performed by the recursion are charged in the same manner as explained above.
  In this way, every reported point is charged $\bigO(1/B)$ I/Os overall. The
  total query time is therefore $\bigO(1 + k/B)$.
\end{proof}

To answer a top-open query with $Q = [\alpha_1, \alpha_2] \times [\beta, U]$,
where $\alpha_1, \alpha_2, \beta \in [U]$, we first identify the chunks $z_1$
and $z_2$ that cover $\alpha_1$ and $\alpha_2$, respectively. This takes
$\bigO(1)$ I/Os by dividing $\alpha_1$ and $\alpha_2$ by the chunk size
$\lambda$, respectively. If $z_1 = z_2$, the query can be solved by searching
the few-point structure of $z_1$ in $\bigO(1 + k/B)$ I/Os
(Lemma~\ref{lmm:div-skysmall}). The subsequent discussion considers $z_1 \neq
z_2$.

\begin{figure}
	\centering
    \yufeigraphics{height=45mm}{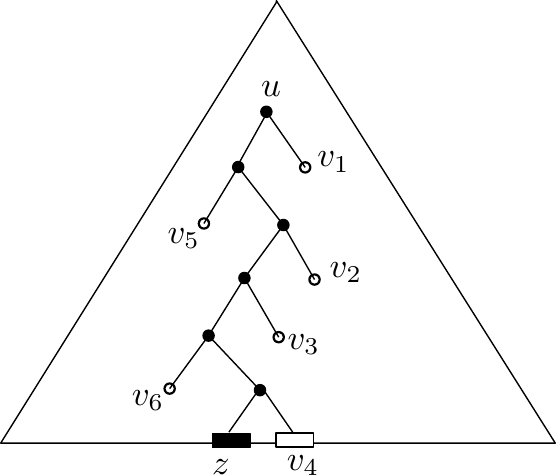}
	\figcapup
	\caption{\boldmath Illustration of $MAX(u)$, $LMAX(z, u)$, and $RMAX(z, u)$}
	\label{fig:div-max}
	\figcapdown
\end{figure}

Let $u$ be the lowest common ancestor of $z_1$ and $z_2$ in $\T$. As $\T$ is a
{\em complete} binary tree, $u$ can be determined in constant I/Os. The rest of
the algorithm proceeds in 4 steps:

\begin{enumerate}
  \item Use the few-point structure of $z_2$ to report the skyline of $P(z_2)
    \intr Q$. Let $S(z_2)$ be the set of points retrieved, and $\beta^*$ the
    $y$-coordinate of the highest point in $S(z_2)$. If $S(z_2) = \emptyset$,
    $\beta^* = \beta$.

  \item Report the set $S_2$ of points in $\lmax(z_2, u)$ whose $y$-coordinates
    are above $\beta^*$. Denote by $v_1, ..., v_c$ the nodes of $\Pi_\ell(z_2,
    u)$ for some integer $c$. For each $i \in [1, c]$, check whether
    $|\high(v_i) \intr S_2| = B$. If not, the subtree of $v_i$ can be
    eliminated. Otherwise, apply Lemma~\ref{lmm:div-half} to retrieve the
    skyline of $P(v_i, \beta_i)$, where $\beta_i$ is the $y$-coordinate of the
    point just to the right of $\highend(v_i)$ in the staircase of $S_2$; if no
    such point exists, $\beta_i = \beta^*$. If $S_2 \neq \emptyset$, update
    $\beta^*$ to be the $y$-coordinate of the highest point in $S_2$.

  \item Find the set $S_1$ of points in $\rmax(z_1, u)$ whose $y$-coordinates
    are above $\beta^*$. Denote by $v'_1, ..., v'_{c'}$ the nodes of
    $\Pi_\gamma(z_1, u)$ for some integer $c'$. For each $i \in [1, c']$, if
    $|\high(v'_i) \intr S_1| = B$, apply Lemma~\ref{lmm:div-half} to retrieve
    the skyline of $P(v'_i, \beta'_i)$, where $\beta'_i$ is the $y$-coordinate
    of the point just to the right of $\highend(v'_i)$ in the staircase of
    $S_1$ (if no such point exists, $\beta'_i = \beta^*$). If $S_1 \neq
    \emptyset$, set $\beta^*$ to the $y$-coordinate of the highest point in
    $S_1$.

  \item Fetch the skyline of $P(z_1) \intr [\alpha_1, \alpha_2] \times
    [\beta^*, U]$ from the few-point structure of $z_1$.
\end{enumerate}
In the example of Figure~\ref{fig:div-qry}, $c = 4$ and $c' = 3$; the algorithm
first obtains the result points from $z_2$, then from the subtrees of $v_1,
..., v_4$, next from the subtrees of $v_1', ..., v_3'$, and finally from $z_1$.

\begin{figure}
	\centering
    \yufeigraphics{height=50mm}{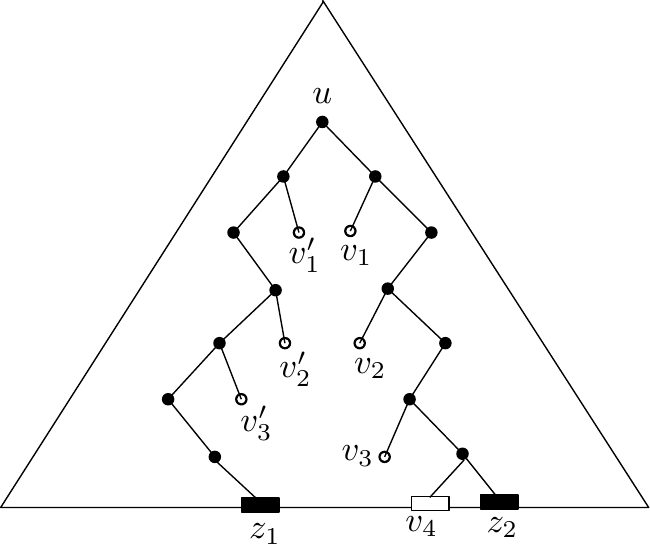}
	\figcapup
	\caption{\boldmath Illustration of the query algorithm}
	\label{fig:div-qry}
	\figcapdown
\end{figure}

To analyze the cost, we focus on the first two steps because the other steps
are symmetric. By Lemma~\ref{lmm:div-skysmall}, Step 1 takes $\bigO(1 + k'/B)$
I/Os, where $k'$ is the number of points reported in this step. In Step 2, by
leveraging the ordering inside $\lmax(z_2, u)$, $S_2$ can be found in $\bigO(1
+ |S_2|/B)$ I/Os. We charge the second term on the points of $S_2$. For each $i
\in [1, c]$, if $\high(v_i) \intr S_2$ has less than $B$ points, the subtree of
$v_i$ incurs no more cost. Otherwise, applying Lemma~\ref{lmm:div-half} takes
$\bigO(k'_i/B)$ I/Os if the application finds $k'_i$ points\footnote{Note that
$k'_i \ge B$ since the whole $high(v_i)$ is definitely reported.}; we charge
this cost on those $k'_i$ points. Overall, every reported point is charged
$\bigO(1/B)$ I/Os. Steps 1-4 each necessitate $\bigO(1)$ extra I/Os. The total
query cost is therefore $\bigO(1 + k/B)$.
\end{fullenv}

\begin{corollary} \label{crl:div-rankmain}
	There is a linear-size structure on a set of $n$ points in $[U]^2$ (where $U
	\ge n$ is an integer) such that a top-open range skyline query can be
	answered optimally in $\bigO(\log \log_B U + k/B)$ I/Os, when $k$
  points are reported.
\end{corollary}
\begin{fullenv}
\begin{proof}
  We simply create the same structure on the input set $P$ of $n$ points in, however, rank space
  $[n]^2$. A query coordinate in $[U]$ can be converted into $[n]$ in
  $O(\log\log_B U)$ I/Os by a standard linear-size structure for predecessor
  search \cite{PT06}. The overall query cost is therefore $\bigO(\log \log_B U
  + k/B)$.

  The optimality follows directly from the reduction explained in the proof of
  Theorem~\ref{thm:topopen-main} and the $\Omega(\log \log_B U)$ lower bound of
  predecessor search under the linear space budget \cite{PT06}.
\end{proof}
\end{fullenv}

\section{Dynamic Top-Open Structure} \label{sec:dynamic}

In this section, we present a dynamic data structure, which is SABE, that uses
linear space, and supports top-open queries in $\bigO(\log_{2B^\epsilon} (n/B) +
k/B^{1-\epsilon})$ I/Os and updates in $\bigO(\log_{2B^\epsilon} (n/B))$ I/Os,
for any parameter~$0\leq \epsilon \leq 1$. We are inspired by the approach of
Overmars and van Leeuwen~\cite{OL81} for maintaining the planar skyline in the
pointer machine. As a brief review, a dynamic binary base tree indexes the
$x$-coordinates of $P$, and every internal node stores the skyline of the points
in its subtree using a secondary search tree. More specifically, the skyline of
an internal node is $(L\backslash L') \cup R$, where $L$ (resp.\ $R$) is the
skyline of its left (resp.\ right) child node, and $L'$ is the set of points in
$L$ dominated by the leftmost (and thus also highest) point of~$R$.

Our approach is based on I/O-CPQAs, which are described in
Section~\ref{sec:iocpqa}.  We observe that attrition can be utilized to maintain
the internal node skylines in~\cite{OL81}, after mirroring the $y$-axis. To
explain this, let us first map the input set $P$ to its mirrored counterpart
$\attr{P} = \{(x_p, -y_p) \mid (x_p, y_p) \in P\}$. In the context of PQAs, we
will interpret each point $(\attr{x}_p, \attr{y}_p)\in \attr{P}$ as an {\em
element} with ``key'' value $\attr{y}_p$ that is inserted at ``time''
$\attr{x}_p$. To formalize the notion of time, we define the $<_x$-ordering of
two elements $\attr{p}, \attr{q} \in \attr{P}$ to be $\attr{p} <_x \attr{q}$, if
and only if $\attr{x}_p < \attr{x}_q$ holds. It is easy to see that element
$\attr{p} \in \attr{P}$ is attrited by element $\attr{q}\in \attr{P}$, if and
only if point $p\in P$ is dominated by point $q \in P$. See
Figure~\ref{fig:AtrittionSkyline} for a geometric illustration of the mirroring
transformation and the effects of attrition.

  \begin{figure}[htb]
    \centering
    \def\svgwidth{0.6\linewidth}
    \executeiffilenewer{./figure/AtrittionSkyline.svg}{./figure/AtrittionSkyline.pdf}
     {inkscape -z -D --file=./figure/AtrittionSkyline.svg
     --export-pdf=./figure/AtrittionSkyline.pdf --export-latex}
    
    \arxivexcl{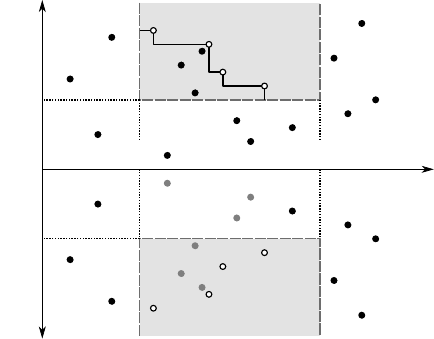}
    {\import{./figure/}{AtrittionSkyline.pdf_tex}}
    \vspace{-0.35cm}
    \caption{The skyline
problem (above) mirrored to the attrition problem (below). White points are
reported for the gray query area $[x_1, x_2] \times [y, \infty[$,
while gray elements are attrited within $[x_1, x_2]$.}
    \label{fig:AtrittionSkyline}
  \end{figure}

Thus, we index the $<_x$-ordering of $\attr{P}$ in a $(2B^{\epsilon},
4B^{\epsilon})$-tree, for a parameter $0\leq \epsilon \leq 1$, and employ
I/O-CPQAs as secondary structures, such that the I/O-CPQA at an internal node
is simply the concatenation of its children's I/O-CPQAs. To obtain logarithmic
query and update I/Os, this sequence of consecutive \textsc{CatenateAndAttrite}
operations at an internal node must be performed in $\bigO(1)$ I/Os
(Lemma~\ref{lem:seq_concats}). The presented I/O-CPQAs are \emph{ephemeral}
(not persistent), and thus the supported operations are \emph{destructive}, as
they destroy the initial configuration of the structure. This only allows
operating on the I/O-CPQA that is the final result of all concatenations and
resides at the root of the base tree. However, in order to support top-open
queries efficiently, accessing I/O-CPQAs at the internal nodes is required.
This is made possible by non-destructive operations. Therefore, we render the
I/O-CPQAs confluently persistent by merely replacing the catenable deques,
which are used as black boxes in our ephemeral construction, with real-time
purely functional catenable deques~\cite{KT99}. Since the imposed overhead is
$\bigO(1)$ worst case I/Os, confluently persistent I/O-CPQAs ensure the same
I/O bounds as their ephemeral counterparts. Section~\ref{sec:skyline} describes
our dynamic data structure in detail.

\subsection{I/O-Efficient Catenable Attrition Priority Queues}
\label{sec:iocpqa}

Here we present ephemeral \emph{I/O-efficient catenable priority queues with
attrition (I/O-CPQAs)} that store a set of elements from a total order and
support all operations in~$\bigO(1)$ I/Os. Also the operations take~$\bigO(1/b)$
amortized I/Os, when a constant number of blocks are already loaded into main
memory for every root I/O-CPQA, for any parameter~$1 \leq b \leq B$. We call
these preloaded records \emph{critical records}. For the sake of simplicity, we
identify an element with its value. Denote by $Q$ an I/O-CPQA and by $\min(Q)$
the smallest element stored in $Q$. We denote by $Q$ also the set of elements in
I/O-CPQA $Q$. Next, we re-state the supported operations in the context of
I/O-CPQAs:

\begin{itemize}
  \item \textsc{FindMin}($Q$) returns~$\min(Q)$.
\vspace{-2mm}
  \item \textsc{DeleteMin}($Q$) returns~$\min(Q)$ and removes it from~$Q$. The
  resulting I/O-CPQA is~$Q' = Q \backslash\{\min(Q)\}$, and~$Q$ is discarded.
\vspace{-2mm}
  \item \textsc{CatenateAndAttrite}($Q_1,Q_2$)\footnote{
    \textsc{InsertAndAttrite}($Q,e$) corresponds to
  \textsc{CatenateAndAttrite}($Q_1,Q_2$), where $Q_2$ contains only element
  $e$.} catenates I/O-CPQA $Q_2$ to the end of another I/O-CPQA~$Q_1$,
  removes all elements in~$Q_1$ that are larger than or equal to~$\min(Q_2)$
  (attrition), and returns the result as a combined I/O-CPQA $Q'_1 = \{e \in
  Q_1 \mid e < \min(Q_2)\} \cup Q_2$. The old I/O-CPQAs $Q_1$ and $Q_2$ are
  discarded.
\end{itemize}

An I/O-CPQA~$Q$ consists of two sorted buffers, called the first buffer $F(Q)$
with $[b,4b]$ elements and the last buffer~$L(Q)$ with $[0,4b]$ elements,
and~$k_Q+2$ deques of records, called the \emph{clean} deque~$C(Q)$, the
\emph{buffer} deque~$B(Q)$ and the \emph{dirty}
deques~$D_1(Q),\ldots,D_{k_Q}(Q)$, where~$k_Q \geq 0$. A \emph{record}~$r =
(l,p)$ consists of a buffer~$l$ of~$[b,4b]$ sorted elements and a pointer~$p$ to
an I/O-CPQA. A record is \emph{simple} when its pointer~$p$ is \nil. The
definition of I/O-CPQAs implies an underlying tree structure when pointers are
considered as edges and I/O-CPQAs as subtrees. We define the ordering of the
elements in a record~$r$ to be all elements of its buffer~$l$ followed by all
elements in the I/O-CPQA referenced by pointer~$p$. We define the queue order of
I/O-CPQA~$Q$ to be~$F(Q)$, $C(Q)$,~$B(Q)$ and~$D_1(Q),\ldots,D_{k_Q}(Q)$ and
$L(Q)$. It corresponds to an Euler tour over the tree structure. See
Figure~\ref{fig:Overview} for an overview of the structure.

\fullcmt{

  \begin{figure}[htb]
    \centering
    \def\svgwidth{\linewidth}
    \executeiffilenewer{./figure/CPQAOverview.svg}{./figure/CPQAOverview.pdf}
     {inkscape -z -D --file=./figure/CPQAOverview.svg
     --export-pdf=./figure/CPQAOverview.pdf --export-latex}
    
    \arxivexcl{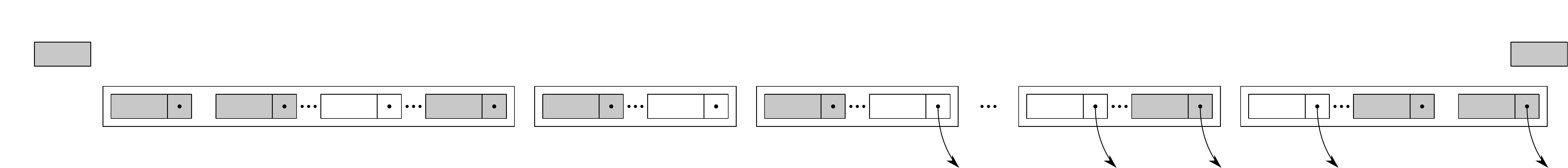}
    {\import{./figure/}{CPQAOverview.pdf_tex}}
    \vspace{-0.75cm}
    \caption{The records in~$C(Q)$
  and~$B(Q)$ are simple, the records of~$D_1(Q),\ldots, D_{k_Q}(Q)$ may contain
  pointers to other I/O CPQAs. I/O-CPQAs imply a tree structure. Gray records
  are critical.}
    \label{fig:Overview}
  \end{figure}

}

\confcmt{
  
  \begin{figure}[htb]
    \centering
    \def\svgwidth{\linewidth}
    \executeiffilenewer{./figure/CPQAOverview.svg}{./figure/CPQAOverview.pdf}
     {inkscape -z -D --file=./figure/CPQAOverview.svg
     --export-pdf=./figure/CPQAOverview.pdf --export-latex}
    
    \arxivexcl{\import{.}{CPQAOverview.pdf_tex}}
    {\import{./figure/}{CPQAOverview.pdf_tex}}
    \vspace{-0.75cm}
    \caption{I/O-CPQA $Q$. Critical
  records are shown in gray.}
    \label{fig:Overview}
  \end{figure}

}

Given a record $r =(l,p)$, the minimum and maximum elements in the buffers
of~$r$, are denoted by~$\min(r) = \min(l)$ and~$\max(r) = \max(l)$,
respectively. They appear respectively first and last in the queue order of~$l$,
since the buffer of~$r$ is sorted by value. Given a deque~$q$, the first and the
last records are denoted by~$\first(q)$ and~$\last(q)$, respectively.
Also,~$\rest(q)$ denotes all records of the deque~$q$ excluding the
record~$\first(q)$. Similarly,~$\front (q)$ denotes all records of the deque~$q$
excluding the record~$\last(q)$. The size $|F(Q)|$ ($|L(Q)|$) of the buffer
$F(Q)$ ($L(Q)$) is defined to be the number of elements in $F(Q)$ ($L(Q)$). The
size~$|r|$ of a record~$r$ is defined to be the number of elements in its
buffer. The size~$|q|$ of a deque~$q$ is defined to be the number of records it
contains. The size~$|Q|$ of the I/O-CPQA~$Q$ is defined to be the number of
elements (both attrited and non-attrited) that~$Q$ contains. For an I/O-CPQA $Q$
we denote by $\first(Q)$ and $\last(Q)$, respectively the first and last records
out of all the records of all the deques $C(Q), B(Q), D_1(Q), \ldots,
D_{k_Q}(Q)$ that exist in $Q$. For an I/O-CPQA~$Q$ we maintain the following
invariants:

\begin{enumerate}[label=I.\arabic*)]
  \item \label{in:records} For every record~$r = (l,p)$ where pointer~$p$
  references I/O-CPQA~$Q'$,~$\max(l) < \min(Q')$ holds.

  \item \label{in:recordpairs} In all deques of~$Q$ where
    record~$r_1=(l_1,p_1)$ precedes record $r_2=(l_2,p_2)$: $\max(l_1) <
    \min(l_2)$ holds.
  \item \label{in:queuevalues} \sloppypar{For the buffer $F(Q)$ and deques
    $C(Q),B(Q),D_1(Q)$: $\max(F(Q)) < \min(\first(C(Q))) < \max(\last(C(Q))) <
  \min(\first(B(Q))) < \min(\first(D_1(Q)))$ holds.}

  \item \label{in:min}  Element~$\min(\first(D_1(Q)))$ is
    the smallest element in the dirty deques~$D_1(Q), \ldots,$ $D_k(Q)$.

    \item  \label{in:smalltail} $\min(\first(D_1(Q)))< \min (L(Q))$.

  \item \label{in:simple} All records in the deques~$C(Q)$ and~$B(Q)$ are
  simple.

  \item \label{in:ineq} $|C(Q)| \geq \sum_{i=1}^{k_Q}{|D_i(Q)|}+k_Q$.

  \item \label{in:small} $|F(Q)|<b$ holds iff $|Q|< b$ holds.

  \item \label{in:noFinChild} If $Q$ is a child of another I/O-CPQA then $F(Q) =
    \emptyset$ and $L(Q) = \emptyset$ holds.
\end{enumerate}
\sloppypar{From
invariants~\iref{in:recordpairs},~\iref{in:queuevalues},~\iref{in:min}
and~\iref{in:smalltail}, we have that~$\min(Q) = \min(F(Q))$.} We say that an
operation \textit{improves} or \textit{aggravates}  the inequality of
Invariant~\iref{in:ineq} by a parameter~$c$ for I/O-CPQA~$Q$, when the
operation, respectively, increases or decreases by $c$ the \textit{state} of
$Q$:
\[
  \Delta (Q) = |C(Q)| - \sum_{i=1}^{k_Q}{|D_i(Q)|} - k_Q
\]

To argue about the~$\bigO(1/b)$ amortized I/O bounds we need more definitions.
\confcmt{The \textit{critical records} of I/O-CPQA $Q$ are $\first(C(Q))$,
  $\first(\rest(C(Q))),$ $\last(C(Q))$, $\first(B(Q))$,  $\first(D_1(Q))$,
  $\last(D_{k_{Q}}(Q))$ and $\last(\front(D_{k_{Q}}(Q)))$, if it exists.
Otherwise $\last(D_{k_{Q}-1}(Q))$ is critical.}
 By~$\records(Q)$ we denote all records in~$Q$ and the records in the I/O-CPQAs
pointed to by $Q$ and its descendants. We call an I/O-CPQA~$Q$ \emph{large} if
$|Q| \geq b$ and \emph{small} otherwise. We define the following potential
functions for large and small I/O-CPQAs. In particular, for large I/O-CPQAs~$Q$
the potential~$\Phi(Q)$ is defined as
\[
  \Phi(Q) = \Phi_F(|F(Q)|) + |\records(Q)| + \Phi_L(|L(Q)|),
\]
where
\begin{eqnarray}
      \Phi_F(x) = \left\{
      \begin{array}{cl}
        5 -\frac{2x}{b}, & b \leq x < 2b \\
        1, & 2b \leq x < 3b \\
        \frac{2x}{b} - 5, & 3b \leq x \leq 4b \\
      \end{array}
    \right. \nonumber
\end{eqnarray}
and
\begin{eqnarray}
      \Phi_L(x) = \left\{
      \begin{array}{cl}
        \frac{x}{b}, & 0 \leq x < b \\
        1, & b \leq x \leq 3b \\
        \frac{2x}{b} - 5, & 3b < x \leq 4b \\
      \end{array}
    \right. \nonumber
\end{eqnarray}

For small I/O-CPQAs~$Q$, the potential~$\Phi(Q)$ is defined as
\[
  \Phi(Q) = \frac{3|Q|}{b}
\]
The total potential $\Phi_T$ is defined as
\[
  \Phi_T = \sum_{Q}{\Phi(Q)} + \sum_{Q, b \leq |Q|}{1},
\]
where the first sum is the total potential of all I/O-CPQAs~$Q$ and the second
sum counts the number of large I/O-CPQAs~$Q$.

\extraspacing \textbf{Operations.} In the following, we describe the algorithms
that implement the operations supported by the I/O-CPQA~$Q$. Most of the
operations call the auxiliary operations \textsc{Bias}$(Q)$ and
\textsc{Fill}$(Q)$, which we describe last. \textsc{Bias} improves the
inequality of~\iref{in:ineq} for~$Q$ by at least $1$ if $Q$ contains any
records. \textsc{Fill}$(Q)$ ensures \iref{in:small}.

\extraspacing \underline{\textsc{FindMin}($Q$)} returns the value $\min(F(Q))$.

\extraspacing \underline{\textsc{DeleteMin}($Q$)} removes element $e =
\min(F(Q))$ from the first buffer $F(Q)$, calls \textsc{Fill}($Q$) and returns
$e$.

\extraspacing \underline{\textsc{CatenateAndAttrite}($Q_1, Q_2$)} creates a new
I/O-CPQA $Q'_1$ by modifying~$Q_1$ and~$Q_2$, and by calling
\textsc{Bias}($Q'_1$), \textsc{Bias}($Q_2$), \textsc{Fill}($Q'_1$) and
\textsc{Fill}($Q_2$).

If $|Q_1| < b$ holds, then $Q_1$ consists only of the first buffer $F(Q_1)$.
Let $F'(Q_1)$ be the non-attrited elements of $F(Q_1)$, under attrition by
$\min(F(Q_2))$. Prepend $F'(Q_1)$ onto the first buffer $F(Q_2)$ of~$Q_2$. If
this prepend causes $F(Q_2) > 4b$, then we take the last $2b$ elements out of
$F(Q_2)$, make a new record out of them and we prepend it onto the deque
$C(Q_2)$.

If~$|Q_2| < b$ holds, then $Q_2$ only consists of $F(Q_2)$. If $|Q_1| < b$ then
we delete attrited elements in $F(Q_1)$ and append $F(Q_2)$ to $F(Q_1)$. We now
assume that $|Q_1| \geq b$. We have three cases, depending on how much of~$Q_1$
is attrited by~$Q_2$. Let $r = (l,\cdot) = \last(Q_1)$ and let $e = \min(Q_2)$.

\begin{enumerate}
  \item \label{it:Q1} $e \leq \min(r)$: Delete $r$. We now have four cases:
  \begin{enumerate}[label=\arabic*)]
    \item \label{it:Q1Cs} If $e \leq \min(F(Q_1))$ holds, we discard
      I/O-CPQA~$Q_1$ and set~$Q'_1 = Q_2$.

    \item \label{it:Q1lastCs} Else if $e \leq \max(\last(C(Q_1)))$ holds, we
      prepend $F(Q_1)$ onto $C(Q_1)$, set $F(Q_1') = \emptyset$, $C(Q_1') =
      \emptyset$, $B(Q_1') = C(Q_1)$, $k_{Q_1'} = 0$ and $L(Q_1') = F(Q_2)$.  We
      call \textsc{Bias}($Q'_1$) once to restore \iref{in:ineq} and then call
      \textsc{Fill}($Q'_1$) once to restore Invariant \iref{in:small}.

    \item \label{it:Bs} Else if $e \leq \min(\first(B(Q_1)))$ or $e \leq
      \min(\first(D_1(Q_1)))$ holds, we set $Q'_1=Q_1$
      and $k_{Q'_1} = 0$ and set $L(Q'_1) =
      F(Q_2)$. If $e \leq \min(\first(B(Q_1)))$ holds, we set~$B(Q_1') =
      \emptyset$, else we set $B(Q_1') = B(Q_1)$.

    \item \label{it:Ds} Else, let $L'(Q_1)$ be the non-attrited elements under
      attrition by $\min(F(Q_2))$. If $|L'(Q_1)| + |F(Q_2)| \leq 4b$ then append
      $F(Q_2)$ to $L'(Q_1)$, else $|L'(Q_1)| + |F(Q_2)| > 4b$ so take the first
      $4b$ elements of $L'(Q_1)$ and $F(Q_2)$ and make into a new record in a
      new last dirty queue of $Q'_1$, leave the rest in $L(Q'_1)$, set $k_{Q'_1}
      = k_{Q_1} + 1$ and call \textsc{Bias}$(Q'_1)$ twice to restore
      \iref{in:ineq}.


  \end{enumerate}

  \item \label{it:Ls} Else if $e \leq \min(L(Q_1))$, we set $Q'_1=Q_1$ and
    $L(Q'_1) = F(Q_2)$.

  \item \label{it:Latts} Else $\min(L(Q_1)) < e$: Let $l'$ be the
    non-attrited elements of $l$, under attrition by $\min(L(Q_1))$,
    and $L'(Q_1)$ be the non-attrited elements, under attrition by
    $e$. If $|L'(Q_1)| + |F(Q_2)| > 4b$ holds, we do the following: if
    $|l'| < |l|$ holds, we put the first $4b - |l'|$ elements of
    $L'(Q_1)$ and $F(Q_2)$ into $l$ along with $l'$. Moreover, if we
    still have more than $3b$ elements left in $L'(Q_1)$ and $F(Q_2)$,
    we put the first $3b$ elements into a new last record of
    $D_{k_{Q_1}}(Q_1)$. Finally, we leave the remaining elements in $L(Q_1)$. If
    we added a new last record to $D_{k_{Q_1}}(Q_1)$, we also  call
    \textsc{Bias}($Q$) once.
\end{enumerate}

We have now entirely dealt with the cases where~$|Q_1| < b$ or $|Q_2| < b$
holds, so in the following we assume that $|Q_1| \geq b$ and $|Q_2| \geq b$
hold, i.e. any I/Os incurred in the cases (\ref{it:Q1C}--\ref{it:D}) below are
already paid for, since the total number of large I/O-CPQAs decreases by one.
Let $e = \min(Q_2)$.

\begin{enumerate}
  \item \label{it:Q1C} If $e \leq \min(F(Q_1))$ holds, we discard
    I/O-CPQA~$Q_1$ and set~$Q'_1 = Q_2$.

  \item \label{it:Q1lastC} Else if $e \leq \max(\last(C(Q_1)))$ holds, we
    prepend $F(Q_1)$ onto $C(Q_1)$ and $F(Q_2)$ onto $C(Q_2)$. We remove the
    simple record $(l,\cdot) = \first(C(Q_2))$ from~$C(Q_2)$, set $Q'_1=Q_1$,
    $F(Q'_1) = \emptyset$, $C(Q'_1) = \emptyset$, $B(Q'_1) = C(Q_1)$, $D_1(Q'_1)
    = (l,p)$, $k_{Q_1'} = 1$, $L(Q'_1) = L(Q_2)$ and $L(Q'_2) = \emptyset$,
    where~$p$ points to~$Q'_2$ if it exists. This gives~$\Delta (Q'_1) = - 2$,
    thus we call \textsc{Bias}$(Q'_1)$ twice and \textsc{Fill}($Q'_1$) once.

  \item \label{it:B} Else if $e \leq \min(\first(B(Q_1)))$ or $e \leq
    \min(\first(D_1(Q_1)))$ holds, we prepend $F(Q_2)$ onto $C(Q_2)$ and remove
    the simple record $(l,\cdot) = \first(C(Q_2))$ from~$C(Q_2)$, set
    $Q'_1=Q_1$, $D_1(Q'_1) =(l,p)$, $k_{Q_1'} = 1$, $L(Q'_1) = L(Q_2)$, $L(Q'_2)
    = \emptyset$ and set~$p$ to point to $Q'_2$, if it exists. If $e \leq
    \min(\first(B(Q_1)))$ holds, we set~$B(Q_1') = \emptyset$, else we set
    $B(Q_1') = B(Q_1)$. This gives~$\Delta (Q'_1) = - 2$ in the worst case, thus
    we call \textsc{Bias}$(Q'_1)$ twice.

  \item \label{it:D} Else let $L'(Q_1)$ be the non-attrited elements of
    $L(Q_1)$, under attrition by $F(Q_2)$. If $|L'(Q_1)| + |F(Q_2)| \leq 4b$
    holds, then we make $L'(Q_1)$ and $F(Q_2)$ into the first record of
    $C(Q_2)$. Else we make them into the first two records of $C(Q_2)$ of size
    $\lfloor (|L'(Q_1)| + |F(Q_2)|)/2 \rfloor$ and $\lceil (|L(Q_1)| +
    |F(Q_2)|)/2 \rceil$ each. We set $Q'_1=Q_1$, $F(Q'_2) = \emptyset$, $L(Q'_1)
    = L(Q_2)$, $L(Q'_2) = \emptyset$, remove $(l_2,\cdot) =\first(C(Q_2))$ from
    $C(Q_2)$. Moreover, we add $(l_2,p)$ as a new single record in
    $D_{k_{Q_1}+1}(Q_1')$, where $p$ points to the rest of $Q'_2$, if it
    exists, and set $k_{Q_1'} = k_{Q_1}+1$. All this aggravates the inequality
    of~\iref{in:ineq} for~$Q'_1$ by at most~$2$, so we call
    \textsc{Bias}$(Q'_1)$ twice.
\end{enumerate}

\extraspacing \underline{\textsc{Fill}$(Q)$} restores Invariant \iref{in:small},
if it is violated. In particular, if $|F(Q)| < b$ and $|Q| \geq b$, let $r =
(l,\cdot) = \first(C(Q))$. If $|l| \geq 2b$ holds, then we take the $b$ first
elements of $l$ and append them to $F(Q)$. Else $|l| < 2b$ holds, so we append
$l$ to $F(Q)$, discard $r$ and call \textsc{Bias}$(Q)$ once.

\extraspacing \underline{\textsc{Bias}$(Q)$} improves the inequality
of~\iref{in:ineq} for~$Q$ by at least~$1$ if $Q$ contains any records. It also
ensures that invariant \iref{in:small} is maintained. We distinguish two
basic cases with respect to $|B(Q)|$, namely $|B(Q)|=0$ and $|B(Q)| > 0$.

\begin{enumerate}
  \item \label{it:Blg0} $|B(Q)| > 0$: We have two cases depending on if $k_Q
    \geq 1$ or $k_Q = 0$.
    \begin{enumerate}[label=\arabic*)]
      \item $k_Q = 0$: Let $e = \min(L(Q))$, if it exists. We remove the first
        record $r_1 =(l_1,\cdot) = \first(B(Q))$ from $B(Q)$. Let~$l_1'$ be the
        non-attrited elements of~$l_1$, under attrition by element~$e$. If
        $|l_1'| = |l_1|$ holds nothing is attrited, so we just add $r_1 =
        (l_1,\cdot)$ at the end of $C(Q)$.

        Else $|l_1'| < |l_1|$ holds, so we set $B(Q) = \emptyset$. If $|l_1'|
        \geq b$ holds, then we make record $r_1$ with buffer $l_1'$ into the new
        last record of $C(Q)$. Else $|l_1'| < b$ holds, so if $|l_1'| + |L(Q)|
        \leq 3b$ also holds, we add $l_1'$ to $L(Q)$ and discard $r_1$. Else
        $|l_1'| + |L(Q)| > 3b$ also holds, so we take the $2b$ first elements of
        $l_1'$ and $L(Q)$ and put them into $r_1$, making it the new last record
        of $C(Q)$.

      \item $k_Q \geq 1$: Let $e = \min(\first(D_1(Q)))$. We remove the first
        record $r_1 =(l_1,\cdot) = \first(B(Q))$ from $B(Q)$. Let~$l_1'$ be the
        non-attrited elements of~$l_1$, under attrition by element~$e$.

        If $|l_1'| = |l_1|$ or $b \leq |l_1'| < |l_1|$ holds, we just add $r_1 =
        (l'_1,\cdot)$ at the end of $C(Q)$ and if $|l_1'| < |l_1|$ we set
        $B(Q)=\emptyset$. Else $|l_1'| < b$ hold, we set $B(Q) = \emptyset$,
        let $r_2 = (l_2,p_2) = \first(D_1(Q))$. If $|l_1'| + |l_2| \leq 4b$
        holds, we discard $r_1$ and prepend $l_1'$ onto $l_2$ of $r_2$. Else
        $|l_1'| + |l_2| > 4b$ holds, so we take the first $2b$ elements of
        $l_1'$ and $l_2$ and put them in $r_1$, making it the new last record of
        $C(Q)$. If this causes $\min(L(Q)) \leq \min(\first(D_1(Q)))$, we
        discard all dirty queues.
    \end{enumerate}
    If $r_1$ was discarded, then we have that $|B(Q)| = 0$ and we call
    \textsc{Bias} recursively, which will not invoke this case again. In
    all cases the inequality of~\iref{in:ineq} for~$Q$ is improved by~$1$.

  \item \label{it:Beq0} $|B(Q)| = 0$: we have three cases depending on the
    number of dirty queues, namely cases $k_Q > 1$, $k_Q = 1$ and $k_Q = 0$.
    \begin{enumerate}[label=\arabic*)]
      \item \label{it:KQgt1} $k_Q > 1$: If $\min(L(Q)) \leq
        \min(\first(D_{k_Q}(Q)))$ holds, we set $k_Q = k_Q-1$ and discard
        $D_{k_Q}(Q)$. This improves the inequality of~\iref{in:ineq} for~$Q$
        by at least~$2$. Else let $e = \min(\first(D_{k_Q}(Q)))$.

        If $e \leq \min(\last(D_{k_Q-1}(Q)))$ holds, we remove the record
        $\last(D_{k_Q -1}(Q))$ from~$D_{k_Q-1}(Q)$. This improves the inequality
        of~\iref{in:ineq} for~$Q$ by~$1$.

      If $\min(\last(D_{k_Q-1}(Q))) < e \leq \max(\last(D_{k_Q-1}(Q)))$ holds,
      we remove record $r_1 = (l_1,p_1) = \last(D_{k_Q-1}(Q))$
      from~$D_{k_Q-1}(Q)$, and let $r_2 = (l_2,p_2) = \first(D_{k_Q}(Q))$. We
      delete any elements in~$l_1$ that are attrited by~$e$, and let~$l_1'$
      denote the set of non-attrited elements. If $|l_1'| + |l_2| \leq 4b$
      holds, we prepend $l_1'$ onto $l_2$ of $r_2$ and discard $r_1$. Else
      we take the first $\lfloor (|l_1'| + |l_2|) / 2 \rfloor$ elements of
      $l_1'$ and $l_2$ and replace $r_1$ of $D_{k_Q-1}(Q)$ with them. Finally,
      we concatenate $D_{k_Q-1}(Q)$ and $D_{k_Q}(Q)$ into a single deque. This
      improves the inequality of~\iref{in:ineq} for~$Q$ by at least~$1$.

    Else $\max(\last(D_{k_Q-1}(Q))) < e$ holds and we just concatenate the
    deques~$D_{k_Q-1}(Q)$ and~$D_{k_Q} (Q)$, which improves the inequality
    of~\iref{in:ineq} for~$Q$ by~$1$.

  \item \label{it:KQeq1} $k_Q = 1$: In this case~$Q$ contains only
    deques~$C(Q)$ and~$D_1(Q)$. Let $r = (l,p) = \first(D_1(Q))$. If
    $\min(L(Q)) \leq \min(\first(\rest(D_1(Q))))$ holds, we discard
    all dirty queues, except for record $r$ of $D_1(Q)$.

    If $\min(L(Q)) \leq \max(l)$ holds, we discard all the dirty deques and let
    $l'$ be the non-attrited elements of $l$. If $|l'| + |L(Q)| \leq 3b$ holds,
    we prepend $l'$ onto $L(Q)$. Else $|l'| + |L(Q)| > 3b$ holds, so we take the
    first $2b$ elements of $l'$ and $L(Q)$ and make them the new last record of
    $C(Q)$ and leave the rest in $L(Q)$. This improves the inequality
    of~\iref{in:ineq} for~$Q$ by $1$.

    Else $\max(\ell) < \min(L(Q))$ holds, so we remove~$r$ and insert buffer~$l$
    into a new record at the end of~$C(Q)$. This improves the inequality
    of~\iref{in:ineq} for~$Q$ by at least~$1$. If~$r$ is not simple, let the
    pointer~$p$ of~$r$ reference I/O-CPQA~$Q'$. We restore \iref{in:simple}
    for~$Q$ by \textit{merging} I/O-CPQAs~$Q$ and~$Q'$ into one I/O-CPQA; see
    Figure~\ref{fig:Bias}. In particular, let~$e = \min(\min(\first(D_1(Q))),
    \min(L(Q)))$.

    We proceed as follows: If $e \leq \min(Q')$ holds, we discard~$Q'$.
    \fullcmt{The inequality of~\iref{in:ineq} for~$Q$ remains unaffected.} Else
    if $\min(\first(C(Q'))) < e \leq \max(\last (C(Q'))$ holds, we set $B(Q) =
    C(Q')$ and discard the rest of~$Q'$. \confcmt{In both cases, t}\fullcmt{T}he
    inequality of~\iref{in:ineq} for~$Q$ remains unaffected.

    Else if $\max(\last(C(Q')) < e \leq \min(\first(D_1(Q')))$ holds, we
    concatenate the deque~$C(Q')$ at the end of~$C(Q)$. If moreover
    $\min(\first(B(Q'))) < e$ holds, we set $B(Q) = B(Q')$. Finally, we discard
    the rest of~$Q'$. This improves the inequality of~\iref{in:ineq} for~$Q$
    by~$|C(Q')|$.

    Else $\min(\first(D_1(Q'))) < e$ holds. We concatenate the deque~$C(Q')$ at
    the end of~$C(Q)$, we set~$B(Q) = B(Q')$, we set $D_1(Q'), \ldots,
    D_{k_{Q'}}(Q')$ as the first $k_{Q'}$ dirty queues of~$Q$ and we
    set~$D_1(Q)$ as the last dirty queue of~$Q$. This improves the inequality
    of~\iref{in:ineq} for~$Q$ by~$\Delta(Q') \geq 0$, since~$Q'$ satisfied
    Invariant \iref{in:ineq} before the operation.

  \item \label{it:KQeq0} $k_Q = 0$: If all deques are empty, $L(Q) \neq
    \emptyset$ and $|F(Q)| \leq 2b$ hold, we take the first $b$ elements of
    $L(Q)$ and append to $F(Q)$. The inequality of~\iref{in:ineq} for~$Q$
    remains $\Delta(Q) = 0$.
  \end{enumerate}
\end{enumerate}

  \begin{figure}[htb]
    \centering
    \def\svgwidth{\linewidth}
    \executeiffilenewer{./figure/CPQABias.svg}{./figure/CPQABias.pdf}
     {inkscape -z -D --file=./figure/CPQABias.svg
     --export-pdf=./figure/CPQABias.pdf --export-latex}
    
    \arxivexcl{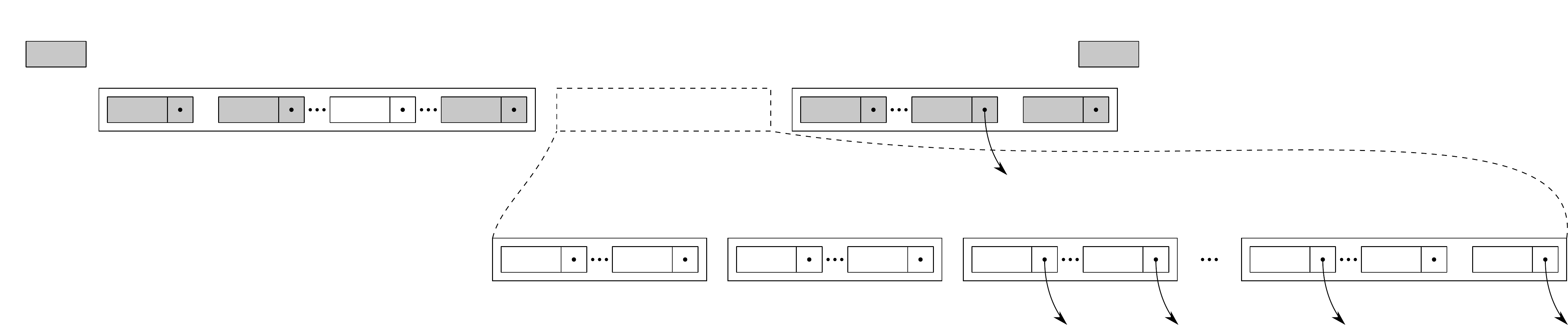}
    {\import{./figure/}{CPQABias.pdf_tex}}
    \vspace{-0.75cm}
    \caption{Merging I/O-CPQAs $Q$ and $Q'$.
This case can only occur when $B(Q)=\emptyset$ and $k_Q=1$.}
    \label{fig:Bias}
  \end{figure}

\begin{theorem} \label{thm:iocpqa}
  An I/O-CPQA supports \textsc{FindMin}, \textsc{DeleteMin},
  \textsc{CatenateAndAttrite} and \textsc{InsertAndAttrite} in~$\bigO(1)$ I/Os
  per operation. It occupies $\bigO((n-m)/B)$ blocks after calling
  \textsc{CatenateAndAttrite} and \textsc{InsertAndAttrite}~$n$ times and
  \textsc{DeleteMin}~$m$ times, respectively.

  All operations are supported by a set of $\ell$ I/O-CPQAs in~$\bigO(1/b)$
  amortized I/Os, when~$M =\Omega ( \ell b)$, using $\bigO((n-m)/b)$ blocks of
  space, for any parameter $1 \leq b \leq B$.
\end{theorem}
\begin{proof}
The correctness follows by closely noticing that we maintain invariants
\iref{in:records}--\iref{in:noFinChild}, which in turn imply that
\textsc{DeleteMin}$(Q)$ and \textsc{FindMin}$(Q)$ always return the minimum
element of~$Q$.
The~$\bigO(1)$ worst case I/O bound is trivial as every operation only
accesses~$\bigO(1)$ records. Although \textsc{Bias} is recursive, notice that in
the case where $|B(Q)| > 0$, \textsc{Bias} only calls itself after making
$|B(Q)| = 0$, so it will not end up in this case again.
We elaborate on all the operations that modify the I/O-CPQA in order to argue
for the amortized bounds:

\begin{confenv}
\noindent \underline{\textsc{DeleteMin}}: If after the call $|F(Q)| \geq b$
holds, no I/Os are incurred and the amortized cost of $\leq \frac{3}{b}$ pays
for increasing the potential. Otherwise $\Phi_F(|F(Q)|) \geq 3$ pays for any
I/Os to call \textsc{Fill} and \textsc{Bias}.

\noindent \underline{\textsc{CatenateAndAttrite}}: $|Q_1| < b$: If $|F'(Q_1)| +
|F(Q_2)| \leq 4b$ holds, $\Phi(|F(Q_1)|)$ pays for any increase in potential.
Else the new record of $C(Q_2)$ is paid for by  $\Delta ( \Phi_T )>1$.

$|Q_2| < b$: In cases (1) and (2) the potential decreases.  In case (3), the
potential does not change if $|L'(Q_1)| + |F(Q_2)| > 4b$. If $L'(Q_1)$ and
$F(Q_2)$ still contain $>3b$ elements, the change in potential is paid for by
$\Delta ( \Phi_T )>0$.

In the following cases, both $Q_1$ and $Q_2$ are large.  Since concatenating
them decreases by one  the number of large I/O-CPQA's, the potential decreases
by at least~$1$, which is enough to pay for any other I/Os also incurred by
\textsc{Bias} and \textsc{Fill}. So we only need to argue that the potential
does not increase in any of the cases.  In fact, in cases (1 - 3) the potential
only decreases.  In case (4), if we make one record, it is paid for by
$\Phi_F(|F(Q_2)|) \geq 1$. Otherwise the second record is paid for by
$\Phi_L(|L(Q_1)|) \geq 1$ if moreover $|L'(Q_1)| + |F(Q_2)| > 4b$ holds, or by
$\Phi_L(|L(Q_1)|) + \Phi_F(|F(Q_2)|) >2$ otherwise.

All I/Os in \textsc{Fill} and \textsc{Bias} have been paid for by a decrease in
potential caused by their caller. Thus, it suffices to argue that these
operations do not increase the potential.

\noindent \underline{\textsc{Fill}}: Indeed, $\Phi_F(|F(Q)|)$ only decreases,
when $|F(Q)| < b$ and $|Q| \geq b$ hold.

\noindent \underline{\textsc{Bias}}: Indeed, cases (1) and (2.1) do not create
new records.  Similarly for (2.2), unless $|l'| + |L(Q)| \leq 3b$ holds, where
$r$ pays for increasing the potential by $1$.  In (2.3) $\Phi_F(|F(Q)|)$ or
$\Phi_L(|L(Q)|)$ decreases.
\end{confenv}
\begin{fullenv}

\noindent \underline{\textsc{DeleteMin}}: If $|F(Q)| \geq b$ holds after
deleting $\min(F(Q))$, then no I/Os are incurred and we only pay an amortized
cost of $\leq \frac{3}{b}$ for increasing the potential. Else $|F(Q)| = b-1$
holds, so $\Phi_F(|F(Q)|) \geq 3$ also holds, which pays for any I/Os in calling
\textsc{Fill} and \textsc{Bias}.

\noindent \underline{\textsc{CatenateAndAttrite}}: If $|Q_1| < b$ holds, then we
prepend the non-attrited elements $F'(Q_1)$ onto $F(Q_2)$. So if $|F'(Q_1)| +
|F(Q_2)| \leq 4b$ holds, then each element of $F(Q_1)$ has a potential of $3/b$,
which is higher than the potential for each element in $\Phi_F$. Thus
$\Phi(|F(Q_1)|)$ pays for any increase in potential. If instead $|F'(Q_1)| +
|F(Q_2)| > 4b$ holds, then $|F(Q_2)| > 3b$ holds, so
\begin{eqnarray*}
  \Delta \Phi_T
    &=&
  \( \frac{3|F(Q_1)|}{b} + \Phi_F(|F(Q_2)|) \) - \( 1+1 \) \\
    &\geq&
  \frac{|F'(Q_1)|}{b} + \frac{2(|F'(Q_1)| + |F(Q_2)|)}{b} - 7 > 1
\end{eqnarray*}
which pays for making the new first record of $C(Q_2)$.

If $|Q_2| < b$ holds, then we have three cases depending on how much of $Q_1$ is
attrited by $Q_2$. Let $e = \min(Q_2)$ and $r = (l,\cdot) = \last(Q_1)$:
\begin{enumerate}
  \item $e \leq \min(\last(D_{k_{Q_1}}(Q_1)))$: We discard $r$ which releases
    $1$ potential and have the four cases:
    \begin{enumerate}[label=\arabic*)]
      \item If $e \leq \min(F(Q_1))$:
        The potential decreases, because we only discard records.

      \item Else if $e \leq \max(\last(C(Q_1)))$: We prepend $F(Q_1)$ onto
        $C(Q)$ and discard records, which only decreases the potential, since
        $\Phi_F(x) \geq 1$ when $x \geq b$. Our calls to \textsc{Bias} and
        \textsc{Fill} are paid for as we discard $r$.

      \item Else if $e \leq \min(\first(B(Q_1)))$ or $e \leq
        \min(\first(D_1(Q_1)))$: We set $L(Q_1) = F(Q_2)$ and discard records,
        which only decreases the potential, since $\Phi_L(x) \leq \Phi_F(x)$ for
        all $x$.

      \item Else: If $|L'(Q_1)| + |F(Q_2)| \leq 4b$ we append $F(Q_2)$ to
        $L'(Q_1)$ and $\Phi_F(|Q_2|)$ pays. Else we make a new dirty queue with
        one new record, which costs 1 potential and 1 potential to cover the
        I/Os in \textsc{Bias}. The total potential difference is
        \begin{eqnarray*}
          \Delta \Phi_T
          &\geq&
          (\Phi_L(|L(Q_1)|) + \Phi(|Q_2|)) - (1+1) \\
          &\geq&
          \(\frac{2(|L'(Q_1)| + |F(Q_2)|)}{b} + \frac{|F(Q_2)|}{b}\) - 7
          \\
          &>& 1
        \end{eqnarray*}
    \end{enumerate}

  \item $e \leq \min(L(Q_1))$: We set $L(Q_1') = F(Q_2)$, which again only
    decreases the potential.

  \item $\min(L(Q_1)) < e$: If $|L'(Q_1)| + |F(Q_2)| > 4b$ holds, then if
    furthermore $|l'| < |l|$ we put the first $4b - |l'|$ elements of
    $L'(Q_1)$, $F(Q_2)$ and $l'$ into $l$, with no change in potential. If there
    are still more than $3b$ elements left in $L'(Q_1)$ and $F(Q_2)$, then we
    put the first $3b$ elements into a new last record of $D_{k_{Q_1}}(Q_1)$ for
    a cost of $1$ in potential and call \textsc{Bias} for a cost of $1$ for
    I/Os, and leave the remaining $\leq 2b$ elements in $L(Q_1)$ for a cost of
    $\leq 1$. All this is paid for, as the total decrease in potential is
    \begin{eqnarray*}
      \Delta \Phi_T
       &\geq&
      \( \Phi_L(|L(Q_1)|) + \Phi_F(|F(Q_2)|) \) - ( 1+1+1 ) \\
       &=&
      \frac{2|L(Q_1)|}{b} + \frac{3|F(Q_2)|}{b} - 8 \\
       &\geq&
      \frac{2(|L'(Q_1)| + |F(Q_2)|)}{b} + \frac{|F(Q_2)|}{b} - 8 > 0
    \end{eqnarray*}
\end{enumerate}
Both $Q_1$ and $Q_2$ are large in all the cases (\ref{it:Q1C}--\ref{it:D}),
hence when we concatenate them, we decrease the potential by at least~$1$, since
the number of large I/O-CPQA's decreases by one, which is enough to pay for any
other I/Os incurred also in \textsc{Bias} and \textsc{Fill}. So we only need to
argue that the potential does not increase in any of the cases.
\begin{enumerate}
  \item If $e \leq \min(F(Q_1))$: the potential decreases, since we discard
    $Q_1$.

  \item Else if $e \leq \max(\last(C(Q_1)))$: we prepend $F(Q_1)$ onto $C(Q_1)$
    and $F(Q_2)$ onto $C(Q_2)$, discard and move around records, which only
    decreases the potential, as $\Phi_F(x) \geq 1$ when $x \geq b$.

  \item Else if $e \leq \min(\first(B(Q_1)))$: we prepend $F(Q_2)$ onto
    $C(Q_2)$, discard and move around records, which only decreases the
    potentials, as $\Phi_F(x) \geq 1$ when $x \geq b$.

  \item Else: We make $L'(Q_1)$ and $F(Q_2)$ into the first one or two records
    of $C(Q_2)$. Since $Q_2$ is large, $|F(Q_2)| \geq b$ holds, and hence we
    have that $\Phi_F(|F(Q_2)|) \geq 1$. If we only make one new record,
    $\Phi_F(|F(Q_2)|)$ pays for it. If we make two records, then $|L'(Q_1)| +
    |F(Q_2)| > 4b$ holds. So if $|L'(Q_1)| \geq b$ moreover holds, then
    $\Phi_L(|L(Q_1)|) \geq 1$ pays for the other record. Else $|L'(Q_1)| < b$
    holds, but then $|F(Q_2)| > 3b$ also holds, so
    \begin{eqnarray*}
      && \Phi_L(|L(Q_1)|) + \Phi_F(|F(Q_2)|) \\
      && =
       \frac{|L(Q_1)|}{b} + \frac{2|F(Q_2)|}{b} - 5 \\
      && \geq
      \frac{|L'(Q_1)|+|F(Q_2)|}{b} + \frac{|F(Q_2)|}{b} - 5
      >
      2
    \end{eqnarray*}
    which pays for both new records.
\end{enumerate}

\noindent \underline{\textsc{InsertAndAttrite}}: The total cost is $\bigO(1/b)$
I/Os amortized, since creating a new I/O-CPQA with only one element and calling
\textsc{CatenateAndAttrite} only costs as much.

\noindent \underline{\textsc{Fill}}: Any I/Os incurred are prepaid by a decrease
in potential made in the procedure calling \textsc{Fill}, so we only need to
argue that the potential does not increase. If $|F(Q)| < b$ and $|Q| \geq b$
then we append at most $2b$ elements to $F(Q)$, hence $\Phi_F(|F(Q)|)$ will only
decrease.

\noindent \underline{\textsc{Bias}}: All I/Os have been paid for by a decrease
in potential caused by the caller of \textsc{Bias}. So we only need to argue
that the potential does not increase because of \textsc{Bias}.
\begin{enumerate}
  \item $|B(Q)| > 0$: We discard, move around and merge records, but we do not
    create new ones. Thus the potential will only decrease.

  \item $|B(Q)| = 0$: We follow the cases of \textsc{Bias}.
  \begin{enumerate}[label=\arabic*)]
    \item $k_Q > 1$: We again discard and move around records, and rearrange
      their elements, but we do not create new records, so the potential will
      only decrease.

    \item $k_Q = 1$: Let $r = (l,p) = \first(D_1(Q))$. If $\min(L(Q)) \leq
      \max(l)$ holds, we might append $l'$ onto $L(Q)$, but only if $|l'| +
      |L(Q)| \leq 3b$.  This will not increase the potential of $L(Q)$ by more
      than $1$, and $r$ pays for that. For the rest of the case we discard and
      move around records and rearrange their elements, but we do not create new
      records, so the potential only decreases.

    \item $k_Q = 0$: If we append the first $b$ elements of $L(Q)$ onto $F(Q)$,
      then $|F(Q)| \leq 2b$ holds, so $\Phi_F(|F(Q)|)$ can only decrease.
      Likewise, when taking at most $b$ elements from $L(Q)$, $\Phi_L(|L(Q)|)$
      will only decrease.
  \end{enumerate}
\end{enumerate}
\end{fullenv}
\end{proof}

\extraspacing\textbf{Catenating a set of I/O-CPQAs.} 
\fullcmt{Define the \textit{state}
of I/O-CPQA $Q$ to be $\Delta(Q)$ and the \textit{critical records} of $Q$ to
be the first three records of $C(Q)$, $\last(C(Q))$, $\first(B(Q))$,
$\first(D_1(Q))$, $\last(D_{k_{Q}}(Q))$ and $\last(\front(D_{k_{Q}}(Q)))$, if
it exists. Otherwise $\last(D_{k_{Q}-1}(Q))$ is critical.} The following lemma
is required by the dynamic structure of the next section.

\begin{lemma} \label{lem:seq_concats}
  A set of I/O-CPQAs $Q_{i}$ for $i\in[1,\ell]$ can be concatenated into a
  single I/O-CPQA without any access to external memory, by calling only
  \textsc{CatenateAndAttrite} operations, provided that for all $i$:
  \begin{enumerate}
    \item $\Delta (Q_{i})\geq 2$ holds, unless $Q_{i}$ contains only one
      record, in which case $\Delta (Q_{i})=0$ or $Q_{i}$ contains only two
      records, in which case $\Delta (Q_{i})=+1$ suffices.

    \item The critical records of $Q_i$ are loaded in main memory.
  \end{enumerate}
\end{lemma}
\begin{fullenv}
\begin{proof}
In fact, the algorithm considers the I/O-CPQAs~$Q_{i}$ in decreasing index~$i$
(from right to left). It first sets $Q^{\ell}=Q_\ell$ and constructs the
temporary I/O-CPQA $Q^{\ell-1}$ by calling
\textsc{CatenateAndAttrite}($Q_{\ell-1}$,$Q^{\ell}$). After the end of the
sequence of operations, the resulting I/O-CPQA~$Q^{1}$ is the concatenation of
all I/O-CPQAs~$Q_{i}$.

To avoid any I/Os during the sequence of \textsc{CatenateAndAttrite}s, we
ensure that~\textsc{Bias} and \textsc{Fill} are not called, and that no more
than the critical records need to be already loaded into memory. 
To avoid calling~\textsc{Bias} we maintain the following invariant during the
sequence of catenations.

\begin{enumerate}[label=I.\arabic*)]
  \setcounter{enumi}{9}
  \item \label{in:seq_cat} Each I/O-CPQAs $Q^{i},i\in[1,\ell]$ constructed
    during the sequence of catenations is in state at least $+1$ unless it
    consists only of the front buffer in which case it is in state $0$.
\end{enumerate}

We prove the invariant inductively on the sequence of operations. Let the
invariant hold for $Q^{i+1}$ and let $Q^{i}$ be constructed by
\textsc{CatenateAndAttrite}($Q_{i}$,$Q^{i+1}$). In the following, we parse the
cases of the \textsc{CatenateAndAttrite} algorithm assumming that
$e=\min(Q^{i+1})$.

If $|Q_i| < b$ holds, then \textsc{Bias} is not invoked and the state of
$Q^{i+1}$ remains $\geq 1$ or is increased by $1$.

If $|Q^{i+1}|<b$ and $|Q_i| \geq b$ then we have to go through the three
respective cases.
\begin{enumerate}
  \item If $e \leq \min(r)$: if record $r$ exists then the state of $Q_i$ is
    increased by $1$ and it becomes $\geq 3$.
  \begin{enumerate}[label=\arabic*)]
    \item If $e \leq \min(F(Q_1))$: Since \textsc{Bias} is not called
      \iref{in:seq_cat} holds trivially.

    \item Else if $e \leq \max(\last(C(Q_1)))$: $Q^i$ is constructed as before
      and we then do the following. Since $k_{Q^{i}}=0$, we take
      out the the first two records of $B(Q^i)$ which are critical since they
      came from $F(Q_i)$ and $\first(Q_i)$. Then, we fill $F(Q^i)$ with one of
      these records provided that no attrition was enforced by $L(Q^i)$. In
      this case, the state of $Q^i$ is $\geq 1$ and the invariant holds. If
      attrition took place then $B(Q^i)$ is discarded and the at most two
      records of $C(Q^i)$ and the record in $L(Q_i)$ are combined (notice that
      all of them are critical) to make $Q^i$ consisting only of records in
      $F(Q^i)$ and $C(Q^i)$ and thus \iref{in:seq_cat} holds.

    \item Else if $e \leq \min(\first(B(Q_1)))$ or $e \leq
      \min(\first(D_1(Q_1)))$: Since \textsc{Bias} is not called 
      \iref{in:seq_cat} holds trivially.

    \item Else: the state at the end is $\geq 0$, since the state of $Q_i$ was
      $\geq 2$ by the induction hypothesis. To restore the invariant that the
      state of $Q^i$ should be $\geq 1$ we check whether
      $\last(D_{k_{Q_i}}(Q^i))$ is attrited or not by the new dirty queue.
      Since both are critical this can be done with no I/Os and thus the state
      of $Q^i$ is increased to $\geq 1$.
  \end{enumerate}

  \item Else if $e \leq \min(L(Q_1))$: since we do not call \textsc{Bias}
    \iref{in:seq_cat} holds trivially.

  \item Else $\min(L(Q_1)) < e$: the state of $Q_i$ is only reduced by $1$ which
    makes the state of $Q^i$ being $\geq 1$ which is sufficient to maintain
    \iref{in:seq_cat}.
\end{enumerate}


Now we move to the more general case where $|Q_1| \geq b$ and $|Q_2| \geq b$.
\begin{enumerate}
  \item $e \leq \min(F(Q_1))$: we do not call \textsc{Bias} so 
    \iref{in:seq_cat} holds trivially.
    
  \item $e \leq \max(\last(C(Q_1)))$: To increase the state of $Q^i$ from $-2$
    to $\geq 1$ we do as follows. We extract the $4$ records of $B(Q^i)$, which
    incurs no I/Os since all four of them are critical (the first was from
    $F(Q_i)$ and the other three from the first $3$ critical records of
    $C(Q_i)$). If no attrition was enforced by $e=\min(Q^{i+1})$, then the state
    of $Q^i$ is $\geq 1$. If attrition is enforced then there are not that many
    records in $B(Q^i)$, then $Q^{i+1}$ is reconstructed (just prepend
    $(l,\cdot)$ to $C(Q^{i+1})$ and then prepend the non-attritted records (at
    most $4$ records) from $Q_i$ to $C(Q^{i+1})$ remaking $F(Q^{i+1})$. At the
    end of this process, the new CPQA $Q^i$ has state at least equal to
    $Q^{i+1}$ which is $\geq 1$ by induction and \ref{in:seq_cat} holds.

  \item $e \leq \min(\first(B(Q_1)))$ or $e \leq \min(\first(D_1(Q_1)))$: we
    will only consider the case where $k_{Q_i}=0$ before the concatenation,
    since otherwise the state of $Q^i$ will be equal or larger to the state of
    $Q_i$, which by the inductive hypothesis is $\geq 2$. Since $Q_i$ must be
    in state $\geq 2$, there are either at least three records in $C(Q_i)$, in
    which case \iref{in:seq_cat} holds and the case is terminated. Otherwise,
    exactly two records exist in $C(Q_i)$ and $B(Q_i)$ is non-empty or there
    are less than two records in $C(Q_i)$ (so the state of $Q_i$ is $\geq 1$ or
    $0$) and $B(Q_i)$ is empty. In the case where two records exist in $C(Q_i)$
    and $B(Q_i)$ is non-empty: if $\first(B(Q_i))$ is not attritted by $e$ we
    put this record into $C(Q_i)$ and now the final I/O-CPQA $Q^{i}$ has state
    $\geq 1$. Otherwise, we restructure $Q^{i+1}$ (as done in the previous case)
    and prepend the non-attrited elements of $Q_i$ onto $Q^{i+1}$ resulting in
    an I/O-CPQA with state at least $\geq 1$ since this was the state of
    $Q^{i+1}$. We follow exactly the same approach in the latter case where
    $C(Q_i)$ contains less than two records and $B(Q_i)$ is empty.

  \item Else: the algorithm works exactly as before with the following
    exception. At the end, $Q^i$ will be in state $\geq 0$, since we added the
    deque $D_{k_{Q^{i+1}}+1}$ with a new record and the inequality of
    \iref{in:ineq} is aggrevated by $2$. To restore the invariant we apply
    Case~\ref{it:Beq0}.  \ref{it:KQgt1} of~\textsc{Bias}. This step requires
    access to records~$\last(D_{k_{Q^i}-1})$ and $\first (D_{k_{Q^i}})$. These
    records are both critical, since the former corresponds to $\last
    (D_{k_{Q^{i+1}}})$ and the latter to $\first(C(Q^{i+1}))$. In addition,
    \textsc{Bias}$(Q^{i+1})$ need not be called, since by the invariant,
    $Q^{i+1}$ was in state $\geq 1$ before the removal of $\first(C(Q^{i+1}))$.
    In this way, we improve the inequality for $Q^i$ by $1$ and
    \iref{in:seq_cat} holds. 
\end{enumerate}
\end{proof}
\end{fullenv}

\subsection{Final Dynamic Top-Open Structure} \label{sec:skyline}

The data structure consists of a base tree, implemented as a dynamic $(a,
2a)$-tree where the leaves store between $k$ and $2k$ elements. We set $a=\lceil
2B^\epsilon\rceil$ and $k=B$, for a given $0 \leq \epsilon \leq 1$.  The base
tree indexes the $<_x$-ordering of $\attr{P}$, and is augmented with confluently
persistent I/O-CPQAs with buffer size $b=B^{1-\epsilon}$ as secondary
structures. In particular, after constructing the base tree, we augment it with
secondary I/O-CPQAs in a bottom-up manner, as follows. For every leaf we make
one I/O-CPQA over its elements, and execute an appropriate amount of
\textsc{Bias} operations, such that the state of the I/O-CPQA satisfies
Lemma~\ref{lem:seq_concats}. We associate the I/O-CPQA with the leaf.  In a
second pass over the leaves, we gather its critical records into a
\textit{representative block} in its parent. The procedure continues one level
above. For every internal node $u$, we access the representative blocks that
contain the critical records of the children I/O-CPQAs of $u$, and
\textsc{CatenateAndAttrite} them into a new I/O-CPQA as implied by
Lemma~\ref{lem:seq_concats}. We execute \textsc{Bias} on the I/O-CPQA enough
times such that its state also satifies Lemma~\ref{lem:seq_concats}.  We
associate the I/O-CPQA with $u$. After the level has been processed, we create
the representative blocks for I/O-CPQAs associated with the nodes of the level,
in the same way as described above. The augmentation ends at the root node of
the base tree. We will ensure that our algorithms access the I/O-CPQA associated
with a node through the representative block stored at the parent of the node.
Thus, it will suffice to explicitly store only the representative blocks in
every internal node and not its associated I/O-CPQA.

Since every leaf contains $\bigO(B)$ elements, the base tree has $\bigO(n/B)$
leaves and thus also $\bigO(n/B)$ internal nodes. Every internal node has
$\Theta(B^\epsilon)$ children, each associated with an I/O-CPQA with $\bigO(1)$
critical records of size $\bigO(B^{1-\epsilon})$. Thus the representative blocks
stored in the internal node occupy $\bigO(1)$ blocks of space. Thus the
structure occupies $\bigO(n/B)$ blocks in total. Assume that $\attr{P}$ is
already sorted by the $<_x$-ordering. The leaves' I/O-CPQAs  are created in
$\bigO(1)$ I/Os, since they contain at most $\bigO(B)$ elements. All
representative blocks are created in $\bigO(n/B)$ I/Os. To create the internal
nodes' I/O-CPQAs, we need only $\bigO(1)$ I/Os to access the representative
blocks and to execute \textsc{Bias} on the resulting I/O-CPQA. Its
representative blocks residing in memory thus are written on disk in $\bigO(1)$
I/Os. Thus the total preprocessing cost is $\bigO(n/B)$ and the structure is
SABE.

\extraspacing\textbf{Updates.} To insert (resp. delete) a point $p$ into (resp.
from)~$P$, we insert (resp. delete) $\attr{p} = (\attr{x}_p, \attr{y}_p)$ in the
structure. In particular, we first find the leaf to insert (resp. delete) that
contains the predecessor of $\attr{x}_p$ (resp. contains $\attr{x}_p$), by a
top-down traversal of the path from the root of the base tree. For every node
$u$ on the path, we also discard the part of its representative block
corresponding to the child that the search path goes into, and $u$'s associated
I/O-CPQA by executing in reverse the operations that created it. Next we insert
(resp. delete) $\attr{p}$ into (from) the accessed leaf, and rebalance the base
tree by executing the appropriate splits and merges on the nodes along the path
in a bottom-up manner. Moreover, we recompute the I/O-CPQA of every accessed
node on the path, as described above. The total update I/Os are
$\bigO(\log_{2B^\epsilon} (n/B))$ in the worst case, since we spend $\bigO(1)$
I/Os to rebalance every accessed node and to recompute its secondary structures.

\extraspacing\textbf{Queries.} To report the skyline points of $P$ that reside
within a given top-open query range $[\alpha_1,\alpha_2] \times [\beta,
\infty[$, we first traverse top-down the two search paths $\widetilde{\pi_1} =
\pi \pi_1$ and $\widetilde{\pi_2} = \pi \pi_2$ from the root of the base tree to
the leaves~$\ell_1$ and~$\ell_2$ that contain points of $\attr{P}$ whose
$<_x$-ordering succeed and precede the query parameters $\alpha_1$ and
$\alpha_2$, respectively. Let node $u$ be on the path $\pi_1 \cup \pi_2$, and
let $c(u)$ be the children nodes of $u$ whose subtrees are fully contained
within $[\alpha_1,\alpha_2]$. For every $u$, we load its representative block
into memory in order to access the critical records of the I/O-CPQAs associated
with $c(u)$ and to \textsc{CatenateAndAttrite} them into a temporary I/O-CPQA,
as implied by Lemma~\ref{lem:seq_concats}. We consider the temporary I/O-CPQAs
of nodes $u$ and the I/O-CPQAs of the leaves~$\ell_1$ and~$\ell_2$ from right to
left, and we \textsc{CatenateAndAttrite} them into one auxiliary I/O-CPQA. The
I/O-CPQAs for $\ell_1$ and $\ell_2$ are created only on the points within the
$x$-range $[\alpha_1,\alpha_2]$ in $\bigO(1)$ I/Os.

To report the skyline points within the query range, we call \textsc{DeleteMin}
on the auxiliary I/O-CPQA. The procedure stops as soon as a point with
$\attr{y}_p > -\beta$ is returned, or when the auxiliary I/O-CPQA becomes empty.

There are $\bigO(\log_{2B^\epsilon} (n/B))$ nodes on $\pi_1 \cup \pi_2$ and we
spend $\bigO (1)$ I/Os to access the representative block of each node. After
this, the construction of the auxiliary I/O-CPQA costs $\bigO(\log_{2B^\epsilon}
(n/B))$ I/Os. Reporting the $k$ output points costs $\bigO
(\frac{k}{B^{1-\epsilon}}+ 1)$ I/Os. Therefore the query takes
$\bigO(\log_{2B^\epsilon}(n/B)+ \frac{k}{B^{1-\epsilon}})$ I/Os in total.  We
conclude that:

\begin{theorem} \label{thm:3sided}
  There is an indivisible linear-size dynamic data structure on $n$ points in
  $\real^2$ that supports top-open range skyline queries in~$\bigO
  (\log_{2B^{\epsilon}} (n/B) + k/B^{1-\epsilon})$ I/Os when $k$ points are
  reported, and updates in~$\bigO(\log_{2B^{\epsilon}} (n/B))$ I/Os for any
  parameter~$0 \leq \epsilon \leq 1$. The structure can be constructed in
  $\bigO(n/B)$ I/Os, assuming an initial sorting on the input points'
  $x$-coordinates.
\end{theorem}

\section{General Range Skyline Queries} \label{sec:4sided}

We now move on to discuss the other variants of range skyline reporting that
are neither symmetric to nor subsumed by top-open queries. It would be nice if
they could be answered in $\bigO(\log_B n + k/B)$ I/Os by a linear-size
structure. Unfortunately, we will prove its impossibility. In fact, even
sub-polynomial query cost is already unachievable for anti-dominance queries,
let alone left-open and 4-sided queries. In fact, anti-dominance, left-open and
4-sided are just as hard as each other. Next, we will formally establish these
facts. \confcmt{Refer to the full version for the proofs.}

\subsection{A Query Lower Bound}

By making a crucial observation on a variant of the low-discrepancy point set
proposed by Chazelle and Liu \cite{CL04}, we manage to prove the next geometric
fact:
\begin{lemma} \label{lmm:4sided-lower}
  For any integer $\omega \ge 1$ and $\lambda \ge 1$, there is a set $P$ of
  $\omega^\lambda$ points in $\real^2$ and a set $G$ of $\lambda
  \omega^{\lambda-1}$ anti-dominance queries such that (i) each query in $G$
  retrieves $d$ points of $P$, and (ii) at most one point in $P$ is returned by
  two different queries in $G$ simultaneously.
\end{lemma}
\begin{fullenv}
\begin{proof}
We first give some definitions in the context of Chazelle and
Liu~\cite{C90,CL04}. A query set~$\mathcal{Q}$ is
$(m,\omega)$-\textit{favorable} for a data set~$S$, if $\forall Q_i \in
\mathcal{Q}: |S \cap Q_i| \geq \omega$ and $\forall i_1 <i_2 \cdots< i_m: |S
\cap Q_{i_1}\cap \cdots \cap Q_{i_m}| = \bigO(1)$. Let~$S$ be a set of~$n$
points in~$\mathbb{R}^2$. Let~$\mathcal{Q} =\{Q_i \subseteq \mathbb{R}^2 | 1
\leq i \leq m \}$ be a set of~$m$ orthogonal 2-sided query ranges~$Q_i =
[q_{i_x}, \infty[ \times [q_{i_y}, \infty[ \subseteq \mathbb{R}^2$. Query range
$Q_i$ is the subspace of $\mathbb{R}^2$ that dominates a given point~$q_i\in
\mathbb{R}^2$ in the positive $x$- and $y$- direction (the ``upper-right''
quadrant defined by~$q_i$). Let $S_i=S \cap Q_i$ be the set of all points
in~$S$ that lie in the range~$Q_i$. An \textit{inverse anti-dominance reporting
query} $Q_i$ contains the points of $S_i$ that do not dominate any other point
in~$S_i$. This problem is equivalent to the anti-dominance problem by inverting
the coordinates of all points and of the query.

We will now construct a $(2,\omega)$-favorable query set $\mathcal{Q}$ and its
corresponding point set $S$, where $\omega>1$. Without loss of generality, we
assume that $n = \omega^\lambda$, where $\lambda>0$, since this restriction
generates a countably infinite number of inputs and thus the lower bound is
general. Let us write $0 \leq i < n$ as~$i=i^{(\omega)}_0 i^{(\omega)}_1 \ldots
i^{(\omega)}_{\lambda-1}$, where~$i^{(\omega)}_j$ is the~$j$-th digit of
number~$i$ in base~$\omega$. Then define
\[
  \rho_{\omega}(i)
  =
  (\omega-i^{(\omega)}_{\lambda-1}-1)(\omega-i^{(\omega)}_{\lambda-2}-1)\ldots
  (\omega-i^{(\omega)}_0-1)
\]
So~$\rho_{\omega}(i)$ is the integer obtained by writing~$0\leq i <n$
using~$\lambda$ digits in base~$\omega$, by first reversing the digits and then
taking their complement with respect to~$\omega$. We define the points of~$S$
to be the set $\{(i,\rho_{\omega}(i))| 0 \leq i <n\}$.\fullcmt{
Figure~\ref{fig:lower} shows an example with $\omega=4$ and $\lambda=2$.}

\fullcmt{

  \begin{figure}[htb]
    \centering
    \def\svgwidth{0.9\linewidth}
    \executeiffilenewer{./figure/LowerBound.svg}{./figure/LowerBound.pdf}
     {inkscape -z -D --file=./figure/LowerBound.svg
     --export-pdf=./figure/LowerBound.pdf --export-latex}
    \tiny
    \arxivexcl{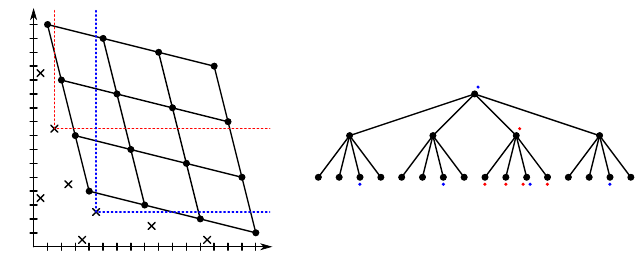}
    {\import{./figure/}{LowerBound.pdf_tex}}
    \normalsize
    \caption{(Left) An example for $\omega
  = 4$ and $\lambda = 2$, the point set $S$ is shown with circles and the the
  queries $\mathcal{Q}$ are shown with crosses. Two examples of queries are
  shown in red and blue. (Right) The corresponding trie that we used to generate
  the point set, here the red and blue queries are also shown, along with the
  internal node which generated the queries.}
    \label{fig:lower}
  \end{figure}

}

To define the query set~$\mathcal{Q}$, we encode the set of points
$\{\rho_\omega(i)|0 \leq i < n\}$ in a full trie structure of depth~$\lambda$.
Recall that $n = \omega^{\lambda}$. Notice that the trie structure is implicit
and it is used only for presentation purposes. Input points correspond to the
leaves of the trie and their $y$ value is their label at the edges of the trie,
where the edges at the root have labels $\omega-i^{(\omega)}_{\lambda-1}-1$ and
the edges at the leafs of the trie have labels $\omega-i^{(\omega)}_0-1$.
Let~$v$ be an internal node at depth~$d$ (namely~$v$ has~$d$ ancestors), whose
prefix~$v_0, v_1, \ldots, v_{d-1}$ corresponds to the path from the root~$r$ of
the trie to~$v$. We take all points in its subtree and sort them by~$y$. From
this sorted list we construct groups of size~$\omega$ by always picking
every~$\omega^{\lambda-d-1}$-th element starting from the smallest non-picked
element for each group. In this case, we say that the query is
\textit{associated} to node~$v$. Each such group corresponds to the output of a
query.\fullcmt{ See Figure~\ref{fig:lower} for an example.}

A node at depth~$d$ has~$\frac{n}{\omega^d}$ points in its subtree and thus it
defines at most~$\frac{n}{\omega^{d+1}}$ queries. Thus, the total number of
queries is:
\[
  \left|\mathcal{Q}\right| = \sum_{d=0}^{\lambda-1}{\omega^d
  \frac{n}{\omega^{d+1}}} = \sum_{d=0}^{\lambda-1}{\frac{n}{\omega}} =
  \frac{\lambda n}{\omega}
\]

In the following we prove that $\mathcal{Q}$ is~$(2,\omega)$-favorable. To
achieve that we need to prove that $\forall Q_i \in \mathcal{Q}: |S \cap Q_i|
\geq \omega$ and $\forall i_1 < i_2: |S \cap Q_{i_1}\cap Q_{i_2}| = \bigO(1)$.

First we prove that we can construct the queries so that they have output size
$\omega$. Assume that we take one of the groups of~$\omega$ points associated
to node~$v$ at depth~$d$. Let the $y$-coordinates of these points be
$\rho_{\omega}(i_1),\rho_{\omega}(i_2),\ldots,\rho_{\omega}(i_{\omega})$ in
increasing order. These have a common prefix of length~$d$ since they all
belong to the subtree of $v$. But we also choose these points so that
$\rho_{\omega}(i_j) -\rho_{\omega}(i_{j-1}) =\omega^{\lambda-d-1}, 1 < j \leq
\omega$. This means that these numbers differ only at the $\lambda - d -1$-th
digit. By inversing the procedure to construct these $y$-coordinates, the
corresponding $x$-coordinates~$i_j, 1 \leq j \leq \omega$ are determined. By
complementing we take the increasing sequence $\bar{\rho}_{\omega}(i_{\omega}),
\ldots, \bar{\rho}_{\omega}(i_2), \bar{\rho}_{\omega}(i_1)$, where
$\bar{\rho}_{\omega}(i_j)=\omega^\lambda-\rho_{\omega}(i_j)-1 $ and
$\bar{\rho}_{\omega}(i_{j-1}) -\bar{\rho}_{\omega}(i_{j})
=\omega^{\lambda-d-1}, 1 < j \leq \omega$. By reversing the digits we finally
get the increasing sequence of $x$-coordinates $i_{\omega},\ldots,i_2,i_1$,
since the numbers differ at only one digit. Thus, the $y$-coordinate of the
group of~$\omega$ points are decreasing as the $x$-coordinates increase, and as
a result a query~$q$ whose horizontal line is just below~$\rho_{\omega}(i_1)$
and the vertical line just to the left of~$\rho_{\omega}(i_{\omega})$ will
certainly contain this set of points in the query. In addition, there cannot be
any other points between this sequence and the horizontal or vertical lines
defining query $q$. This is because all points in the subtree of~$v$ have been
sorted with respect to~$y$, while the horizontal line is positioned just
below~$\rho_{\omega}(i_1)$, so that no other element lies in between. In the
same manner, no points to the left of~$\rho_{\omega}(i_{\omega})$ exist, when
positioning the vertical line of~$q$ appropriately. Thus, for each query~$q \in
\mathcal{Q}$, it holds that~$|S\cap q|=\omega$.

We now want to prove that for any two query ranges $p,q \in \mathcal{Q}$, $|S
\cap q \cap p| \leq 1$ holds. Assume that~$p$ and~$q$ are associated to
nodes~$v$ and~$u$, respectively, and that their subtrees are disjoint. That
is,~$u$ is not a proper ancestor or descendant of~$v$. In this case, $p$
and~$q$ share no common point, since each point is used only once in the trie.
For the other case, assume without loss of generality that~$u$ is a proper
ancestor of~$v$ ($u \neq v$). By the discussion in the previous paragraph, each
query contains~$\omega$ numbers that differ at one and only one digit.
Since~$u$ is a proper ancestor of~$v$, the corresponding digits will be
different for the queries defined in~$u$ and for the queries defined in~$v$.
This implies that there can be at most one common point between these
sequences, since the digit that changes for one query range is always set to a
particular value for the other query range.
\end{proof}
\end{fullenv}

We use the term {\em $(\omega, \lambda)$-input} to refer to the point set $P$
obtained in Lemma~\ref{lmm:4sided-lower} after $\omega$ and $\lambda$ have been
fixed. We deploy such input sets to derive:

\begin{lemma} \label{lmm:4sided-lower-detailed}
  Regarding anti-dominance queries on $n$ points in $\real^2$, any structure
  (in the indexability model) of at most $cn/B$ blocks must incur
  $\Omega((n/B)^{1/(25c)} + k/B)$ I/Os to answer a query in the worst case,
  where $c \ge 1$ is a constant and $k$ is the result size.
\end{lemma}

\begin{fullenv}
\begin{proof}
  Let us first review the {\em indexability theorem} of \cite[Theorem
  5.5]{HKMPS02}. Let $\Lambda$ be a structure on a $(\omega, \lambda)$-input.
  Define the {\em access overhead} of $\Lambda$ as the smallest value $A$ that
  allows us to claim: $\Lambda$ answers any query with output size $\omega$ in
  $A\omega/B$ I/Os. In the context of Lemma~\ref{lmm:4sided-lower}, the
  indexability theorem states:
	\begin{center}
    {\em if $\omega \ge \fr{B}{2}$ and $A \le \fr{\sqrt{B}}{4}$, $\Lambda$ must
    use at least $\fr{\lambda}{12} \fr{\omega^\lambda}{B}$ blocks.}
	\end{center}

  Next, we will argue that if a structure has query complexity
  $\bigO((n/B)^{1/(25c)} + k/B)$, it must use strictly more than $cn/B$ blocks
  in the worst case. This implies that no structure of at most $cn/B$ blocks
  can guarantee the aforementioned query time, and hence, proving
  Lemma~\ref{lmm:4sided-lower-detailed}.

  Consider any structure with query time $\bigO((n/B)^{1/(25c)} + k/B)$. Let
  $\Lambda$ be the structure's instance on an $(\omega, \lambda)$-input where
  $\omega = B$ and $\lambda = 12c+1.1$. The I/O cost of $\Lambda$ answering a
  query with output size $k = \omega$ is at most
	\begin{eqnarray}
		&& \alpha((\omega^\lambda/B)^{1/(25c)} + \omega/B) \nn \\
		&=&
		\alpha B^{\fr{12c + 0.1}{25c}} + \alpha
		=
		\alpha B^{\fr{12}{25} + \fr{0.1}{25c}} + \alpha
		\le
		\alpha B^{\fr{12.1}{25}} + \alpha \nn
	\end{eqnarray}
  where $\alpha > 0$ is a certain constant. It thus follows that $A \le \alpha
  B^{\fr{12.1}{25}} + \alpha < \sqrt{B}/4$ when $B$ is sufficiently large.
  Therefore, by the indexability theorem, the structure must occupy at least
  $(\lambda/12) \omega^\lambda/B = (c + 1.1/12)n/B > cn/B$ blocks.
\end{proof}
\end{fullenv}


\begin{theorem} \label{thm:4sided-lower2}
  Regarding anti-dominance queries on $n$ points, any linear-size structure
  under the indexability model must incur $\Omega((n/B)^\eps + k/B)$ I/Os
  answering a query in the worst case, where $\eps > 0$ can be an arbitrarily
  small constant, and $k$ is the result size.
\end{theorem}

\confcmt{
   \noindent {\bf Remarks.} In the full version, we utilize
   Lemma~\ref{lmm:4sided-lower} to prove that any internal memory pointer-based
   data structure that supports anti-dominance queries in~$\bigO(\log^{\bigO(1)}
   + k)$ time requires $\Omega(n\frac{\log{n}}{\log{\log{n}}})$ space. Thus, the
   dynamic structure of~\cite{BT11} for 4-sided queries occupies optimal space
   within a $\bigO(\log \log n)$ factor, for the attained query time.
 }

\subsection{Query-Optimal Structure}

The above lower bound is tight. In fact, we are able to prove a stronger fact: a
{\em 4-sided} query can be answered in $\bigO((n/B)^\eps + k/B)$ I/Os by a
linear-size dynamic structure. \confcmt{Deferring the details to the full
version, we claim:}

\begin{theorem} \label{thm:4sided-main}
  There is an indivisible linear-size structure on $n$ points in $\real^2$ such
  that, 4-sided range skyline queries can be answered in $\bigO((n/B)^\eps +
  k/B)$ I/Os, where $k$ is the number of reported points. The query cost is
  optimal under the indexability model. The structure can be updated in
  $\bigO(\log (n/B))$ amortized I/Os.
\end{theorem}
\begin{fullenv}
\begin{proof}
\extraspacing \textbf{Structure.} Create a weight-balanced B-tree
\cite{AV03} $T$ on the $x$-coordinates of the points in $P$. Each leaf node of
$T$ has capacity $B$, and each internal node has $\Theta(f)$ child nodes where
$f = (n/B)^\eps / \log (n/B)$. The height $h$ of $T$ is thus $\bigO(\log_f
(n/B)) = \bigO(1)$.  For a node $u$ in $T$, let $P(u)$ be the set of points
whose $x$-coordinates are in the subtree of $u$.  We manage $P(u)$ using a
structure $R(u)$ of Theorem~\ref{thm:3sided} for answering right-open queries.
Specifically, $R(u)$ answers a right-open query and supports an update in
$\bigO(\log(|R(u)|/B))$ I/Os. The right-open structures of all nodes at the
same level of $T$ consume $\bigO(n/B)$ space in total. As $T$ has only constant
levels, the total space cost is $\bigO(n/B)$.

\extraspacing \textbf{Query.} Given a 4-sided query with search
rectangle $Q = [\alpha_1, \alpha_2] \times [\beta_1, \beta_2]$, we find in
$\bigO(h f/B) = \bigO((n/B)^\eps)$ I/Os the leaf nodes $z_1, z_2$ of $T$
containing the successor and predecessor of $\alpha_1$ and $\alpha_2$
respectively, among the $x$-coordinates indexed by $T$. If $z_1 = z_2$, solve
the query by loading the $B$ points in $z_1$ into memory with $\bigO(1)$ I/Os.

Consider now $z_1 \neq z_2$. Let $\pi_1$ ($\pi_2$) be the path from the lowest
common ancestor of $z_1$ and $z_2$ to $z_1$ ($z_2$). Let $S$ be the set of
child nodes $v$ of the internal nodes on $\pi_1 \cup \pi_2$ such that the
$x$-interval of $v$ is fully contained in $[\alpha_1, \alpha_2]$ (the
$x$-interval of $v$ tightly encloses the $x$-coordinates in the subtree of
$v$). The nodes of $S$ have disjoint $x$-intervals, and can be listed out in
descending order of their $x$-intervals with $\bigO(hf/B) = \bigO((n/B)^\eps)$
I/Os. Also, $|S| \le hf = \bigO((n/B)^\eps / \log (n/B))$.

Find the skyline of $P(z_2) \cap Q$ in one I/O; let $\beta^*$ be the
$y$-coordinate of the highest point in this skyline. Next, we process the nodes
of $S$ in descending order of their $x$-intervals. For each $v \in S$, perform
a right-open query with $]-\infty, \infty[ \times [\beta^*, \beta_2]$ on
$R(v)$, and output all the points retrieved. If the query returns at least one
point, update $\beta^*$ to the $y$-coordinate of the highest point returned.
Finally, issue a 4-sided query with $[\alpha_1, \alpha_2] \times [\beta^*,
\beta_2]$ on $z_1$ in one I/O.

Since each right-open query costs $\bigO(\log (n/B))$ I/Os (plus linear output
time), all such queries on the nodes of $S$ have total cost $\bigO(|S| \log
(n/B) + k/B) = \bigO((n/B)^\eps + k/B)$.

\extraspacing \textbf{Update}. To insert a point $p$ into $P$, first
descend a root-to-leaf path $\pi$ to the leaf node $z$ of $T$ where $p_x$
should be placed. For each interval node $u$ along $\pi$, insert $p$ to $R(u)$
in $\bigO(\log (n/B))$ I/Os. Since $T$ has $h = \bigO(1)$ levels, the cost so
far is $\bigO(\log (n/B))$.

Next, update the base tree $T$ by inserting $p_x$. If an internal node $u$ is
split, we construct $R(u')$ for each new node $u'$ from scratch by simply
inserting into $R(u')$ all the relevant points in $\bigO(|P(u)| \log (n/B))$
I/Os. The cost can be charged on the $\Omega(|P(u)|)$ updates that have
occurred beneath $u$ since its creation. Hence, each of those updates bears
$\bigO(\log(n/B))$ I/Os. Since an update needs to bear such cost only $h$
times, the total amortized cost is still $\bigO(\log(n/B))$.

A deletion can be handled in a similar manner. Finally, reconstruct the entire
structure after $\Omega(n)$ updates to make sure that $h$ does not change until
$T$ is rebuilt next time. Standard analysis shows that the amortized update
overhead remains $\bigO(\log(n/B))$.
\end{proof}

\subsection{Pointer Machine Space Lower Bound}

In the pointer machine (PM) model, a data structure that stores a data set $S$
and supports range reporting queries for a query set $\mathcal{Q}$, can be
modelled as a directed graph $G$ of bounded out-degree with some nodes being
\emph{entry nodes}. In particular, every node in $G$ may be assigned an element
of $S$ or may contain some other useful information. For a query range $Q_i\in
\mathcal{Q}$, the algorithm navigates over the edges of $G$ in order to locate
all nodes that contain the answer to the query. The algorithm may also traverse
other nodes. The time complexity of reporting the output of~$Q_i$ is at least
equal to the number of nodes accessed in graph~$G$ for~$Q_i$.

Given a directed graph $G$ modelling a data structure in the PM, Chazelle and
Liu~\cite{C90,CL04} define the graph~$G$ to be
$(\alpha,\omega)$-\textit{effective}, if a query is supported in $\alpha(k +
\omega)$ time, where~$k$ is the output size,~$\alpha$ is a multiplicative
factor for the output size ($\alpha = \bigO(1)$ for our purposes) and~$\omega$
is the additive factor.  Moreover, a query set~$\mathcal{Q}$ is
$(m,\omega)$-\textit{favorable} for a data set~$S$, if $\forall Q_i \in
\mathcal{Q}: |S \cap Q_i| \geq \omega$ and $\forall i_1 <i_2 \cdots< i_m: |S
\cap Q_{i_1}\cap \cdots \cap Q_{i_m}| = \bigO(1)$. Intuitively, the first part
of this property requires that the size of the output is large enough (at
least~$\omega$) so that it dominates the additive factor of~$\omega$ in the
time complexity. The second part requires that the query outputs have minimum
overlap, in order to force~$G$ to be large without many nodes containing the
output of many queries. The following lemma exploits these properties to
provide a lower bound on the minimum size of~$G$.

\begin{lemma}[From {\cite[Lemma 2.3]{CL04}}] \label{lem:lower}
  For an $(m,\omega)$-effective graph~$G$ for the data set~$S$, and for an
  $(\alpha,\omega)$-favorable set of queries~$\mathcal{Q}$, the graph $G$
  contains~$\Omega(|\mathcal{Q}|\omega/m)$ nodes, for constant~$\alpha$ and for
  any large enough~$\omega$.
\end{lemma}

\begin{theorem} \label{thm:lower}
  The anti-dominance reporting problem in the Pointer Machine requires
  $\Omega(n\frac{\log{n}}{\log{\log{n}}})$ space, if the query is
  supported in~$\bigO(\log^\gamma{n} + k)$ time, where~$k$ is the size of the
  answer to the query and parameter $\gamma = \bigO(1)$.
\end{theorem}
\begin{proof}
  Lemma~\ref{lmm:4sided-lower} allows us to apply Lemma~\ref{lem:lower}, when
  setting $\omega=\log^\gamma{n}$ and $\lambda=\left\lfloor
  \frac{\log{n}}{1+\gamma\log{\log{n}}}\right \rfloor$, for some constant
  $\gamma>0$. Thus the query time of $\bigO(\log^\gamma{n} + k)$, for output
  size~$k$, can only be achieved at a space cost of
  $\Omega(n\frac{\log{n}}{\log{\log{n}}})$.
\end{proof}
\end{fullenv}

%


\section*{ACKNOWLEDGEMENTS}

The work of Yufei Tao and Jeonghun Yoon was supported in part by (i) projects GRF 4166/10, 4165/11, and 4164/12 from HKRGC, and (ii) the WCU (World Class University) program under the National Research Foundation of Korea, and funded by the Ministry of Education, Science and Technology of Korea (Project No: R31-30007).

\bibliographystyle{abbrv}
\bibliography{References}

\end{document}